\documentclass[10pt]{article}

\usepackage[authoryear,round]{natbib}                                                   
\usepackage{amsmath}                                                                                                            
\usepackage{graphicx}                                                                                                           
\usepackage{subfigure}                                                                                                          
\usepackage{amssymb}                                                                                                            
\usepackage[mathscr]{eucal}                                                                                             
\usepackage{cancel}                                                                                                                             
\usepackage[normalem]{ulem}                                                                                                                             
\usepackage{pstricks}
\usepackage{rotating}
\usepackage{lscape}
\usepackage[paperwidth=8.5in,paperheight=11in,top=1.25in, bottom=1.25in, left=1.00in, right=1.00in]{geometry}
\usepackage{mathtools}                                                                                                          
\mathtoolsset{showonlyrefs=true}                                                                        
\usepackage{amsthm}                                                                                                                             
\theoremstyle{plain}
\newtheorem{theorem}{Theorem}
\newtheorem{proposition}[theorem]{Proposition}  
      
\theoremstyle{definition}
\newtheorem{definition}[theorem]{Definition}
\newtheorem{assumption}[theorem]{Assumption}

\newtheorem{remark}[theorem]{Remark}

\usepackage{hyperref}

\newcommand{\<}{\langle} 
\renewcommand{\>}{\rangle}
\renewcommand{\(}{\left(}                               
\renewcommand{\)}{\right)}
\renewcommand{\[}{\left[}
\renewcommand{\]}{\right]}                      
\newcommand{\norm}[1]{\left\|{#1}\right\|}              
\newcommand{\dom}[1]{{\rm dom}(#1)}     

\newcommand\Cb{\mathbb{C}}      
\newcommand\Eb{\mathbb{E}}                      
\newcommand\Pb{\mathbb{P}}      
                                                      
\newcommand\Rb{\mathbb{R}}                                                                              
\newcommand\Ib{\mathbb{I}}
\newcommand\BS{\mathrm{BS}}
                                        
\newcommand\Ac{\mathscr{A}}
\newcommand\Bc{\mathscr{B}}

\newcommand\Ec{\mathscr{E}}
\newcommand\Fc{\mathscr{F}}

\newcommand\Hc{\mathscr{H}}

\newcommand\Oc{\mathscr{O}}
\newcommand\Pc{\mathscr{P}}
\newcommand\Sc{\mathscr{S}}

\newcommand\Ad{\Ac^*}

\newcommand\eps{\varepsilon}
\newcommand\om{\omega}
\newcommand\Om{\Omega}
\newcommand\sig{\sigma}

\newcommand\lam{\lambda}
\newcommand\del{\delta}

\newcommand\Nt{\widetilde{N}}

\renewcommand\d{\partial}

\usepackage{color}

\begin{document}

\title{
The Smile of certain L\'evy-type Models
}
\author{ 
Antoine Jacquier
\thanks{Department of Mathematics, Imperial College London, United Kingdom.}
\and
Matthew Lorig
\thanks{ORFE Department, Princeton University, Princeton, USA.  Work partially supported by NSF grant DMS-0739195.}
}

\date{This version: \today}

\maketitle

\begin{abstract}
We consider a class of assets whose risk-neutral pricing dynamics are described by an exponential L\'evy-type process subject to default.  The class of processes we consider features locally-dependent drift, diffusion and default-intensity as well as a locally-dependent L\'evy measure.  Using techniques from regular perturbation theory and Fourier analysis, we derive a series expansion for the price of a European-style option.  We also provide precise conditions under which this series expansion converges to the exact price.  Additionally, for a certain subclass of assets in our modeling framework, we derive an expansion for the implied volatility induced by our option pricing formula.  The implied volatility expansion is exact within its radius of convergence.  
As an example of our framework, we propose a class of CEV-like L\'evy-type models.  Within this class, approximate option prices can be computed by a single Fourier integral and approximate implied volatilities are explicit (i.e., no integration is required).  Furthermore, the class of CEV-like L\'evy-type models is shown to provide a tight fit to the implied volatility surface of S{\&}P500 index options.
\end{abstract}


%
%

\section{Introduction}
\label{sec:intro}
A \emph{local volatility} model is a model in which the volatility $\sig_t$ of an asset $X$ is a function of the current time $t$ and the present level of $X$.  That is, $\sig_t = \sig(t,X_t)$.  One advantage of local volatility models is that, like most scalar diffusions, transition densities (and therefore option prices) are often available in closed-form as eigenfunction expansions (see \citet*{linetskybook,lipton2002forex} and references therein).  However, local volatility models suffer from the fact that they do not permit the underlying asset to experience jumps, the need for which is well-documented in literature \citet*{eraker}.
Furthermore, local volatility models do not account properly for the forward volatility, and notoriously misprice options such as cliquet or forward-start.
\par
One class of models that does allow the underlying to jump is the exponential L\'evy class.  In this class, the underlying $X = e^Y$ is described as the exponential of a L\'evy process $Y$.  Aside from allowing the underlying to jump, exponential L\'evy models have the desirable feature that transition densities (and European option prices) can be computed quickly as generalized Fourier transforms (see \citet*{lewis2001simple,lipton2002,levendorskiibook,contbook}).  However, exponential L\'evy models are spatially homogeneous; neither the drift, volatility nor the jump-intensity have any local dependence.  Thus, exponential L\'evy models are not able to exhibit volatility clustering or capture the leverage effect, both of which are well-known features of equity markets.
\par
Recently, a number of authors have found methods of combining the desirable features of local volatility and exponential L\'evy models.  For example, \citet*{gobet-smart} derive an analytical formula for the approximate prices of European options, for models that include local volatility and compound Poisson jumps (i.e., models that include a finite activity L\'evy measure).   Their approach relies on asymptotic expansions around small diffusion and small jump frequency/size limits.  More recently, \citet*{pascucci} consider general local volatility models with independent L\'evy jumps (not necessarily finite activity).  Unlike, \citet{gobet-smart}, \citet{pascucci} make no small jump intensity/size assumption.  Rather the authors construct an asymptotic solution of the pricing equation by expanding the local volatility function as a Taylor series.  While both of the methods described above allow for local volatility and \emph{independent} jumps, neither of these methods allow for \emph{state-dependent} jumps.
\par
Stochastic jump-intensity is an important feature of equity markets (see \citet*{christoffersen2009}) and a locally dependent L\'evy measure is one way to incorporate stochastic jump-intensity into a modeling framework.  One analytically tractable way of obtaining a local L\'evy measure is to time-change a scalar Markov diffusion with a L\'evy subordinator, as described in \citet*{carr}.  Another analytically tractable method of working with local L\'evy measures is to write a local L\'evy measure as a power series in its local variable, as described in \citet*{lpp}.
\par
In this paper, we take a different approach.  We consider a L\'evy-type process whose infinitesimal generator separates into locally dependent and independent parts.  The locally independent part is the generator of a L\'evy process with killing.  We treat the locally dependent part of the generator as a regular perturbation about the locally independent part.  Thus, we are able to obtain a convergent series representation for the price of a European option.  A significant advantage of this method is that, when the locally independent part of the generator has no jump or killing component, we are able find a convergent series expansion for the implied volatility surface induced by our option pricing formula.
\par
The rest of this paper proceeds as follows: In Section \ref{sec:assumptions}, we present a class of exponential L\'evy-type models and state our assumptions about the market.  In Section \ref{sec:pricing}, using regular perturbation methods and Fourier analysis, we derive a series expansion for the price of a European option.  We also provide precise conditions under which this series converges to give the exact option price.  In Section \ref{sec:impvol}, we provide a series expansion for the implied volatility smile induced by a certain sub-class of models within our modeling framework.  This series is exact within its radius of convergence.  In Section \ref{sec:example}, we perform specific computations for a class of CEV-like L\'evy-type model.  In this class, approximate option prices can be computed by a single Fourier integral; approximate implied volatilities are explicit, requiring no integration.  Section \ref{sec:example} also includes extensive numerical examples, including a calibration to S{\&}P500 options.  Proofs and some sample Mathematica code can be found in an Appendix.  Lastly, some concluding remarks are given in Section \ref{sec:conclusion}.

%
%

\section{Model and assumptions}
\label{sec:assumptions}
We assume a frictionless market, no arbitrage and take an equivalent martingale measure $\Pb$ to be chosen by the market on a complete filtered probability space $(\Om,\Fc,\{\Fc_t,t \geq 0\},\Pb)$.  The filtration $\{\Fc_t, t \geq 0 \}$ represents the history of the market.  All processes defined below live on this space.  For simplicity, we assume zero interest rates and no dividends.  Thus, in the absence of arbitrage, all traded assets are martingales.  We consider a risky asset $X$, whose dynamics are given by
\begin{align}
X_t
        &=              \Ib_{\{t<\zeta\}} \exp(Y_t) , \\
dY_t
        &=              \alpha(Y_t) dt + \sig(Y_t) dW_t + \int_\Rb z d\Nt_t(Y_{t-}, dz), &
Y_0
        &=              y \in \Rb , \label{eq:dY} \\
\zeta
        &=              \inf \left\{ t \geq 0 : \int_0^t k(Y_s) ds = \Ec \right\} , &
\Ec
        &\sim   \text{Exp}(1) ,
\end{align}
where $W$ is a Brownian motion, $\Ec$ is an independent exponentially distributed random variable with parameter one, and $d\Nt_t(Y_{t-}, dz)$ is a state-dependent compensated Poisson random measure
\begin{align}
d\Nt_t(Y_{t-}, dz)
        &=              dN_t(Y_{t-}, dz) - \nu(Y_{t-},dz) , &
\Eb [dN_t(Y_{t-}, dz)|Y_{t-}]
        &=              \nu(Y_{t-},dz) dt .
\end{align}
The volatility, killing, and drift functions, as well as the state-dependent L\'evy measure are given by
\begin{align}
\sig(y)
        &=              \Big( a_0^2 + \eps a_1^2 \eta(y) \Big)^{1/2} , \label{eq:sig} \\
k(y)
        &=              c_0 + \eps c_1 \eta(y) , \label{eq:k} \\
\nu(y,dz)
        &=              \nu_0(dz) + \eps \eta(y) \nu_1(dz) , \label{eq:nu} \\
\alpha(y)
        &=              k(y) - \frac{1}{2} \sig^2(y) - \int_\Rb \nu(y, dz)\Big( e^z - 1 - z \Big) . \label{eq:alpha}
\end{align}
Here, $(a_0,a_1,c_0,c_1,\eps)$ are non-negative constants and the function $\eta$ belongs to $\Sc$, the Schwartz space of rapidly decreasing functions on $\Rb$:
\begin{align}
\Sc
        &=              \{ f \in C^\infty(\Rb): \norm{f}_{\alpha,\beta}< \infty, \ \text{for all }(\alpha, \beta) \in\mathbb{N}^2\} , &
\norm{f}_{\alpha,\beta}
        &:=              \sup_{y \in \Rb} |y^\alpha \d^\beta f(y)| .                             \label{eq:schwartz}
\end{align} 
The function $\eta$ must be such that  
$\sig(y) > 0$, 
$k(y) \geq 0$ and $\nu(y,A) \geq 0$ for all $y \in \Rb$ and all Borel sets $A$.  Finally, we assume that the locally-dependent L\'evy measure $\nu(y,dz)$ satisfies, for any $y\in\Rb$, 
\begin{equation}\label{eq:defLevyMeasure}
\int_\Rb \min(1,z^2) \nu(y,dz) < \infty, \qquad
\nu(y,\{0\}) = 0,
\end{equation}
\begin{equation}\label{eq:PropLevyMeasure}
\int_{|z| \geq 1} e^z \nu(y,dz) < \infty, \qquad
\int_{|z| \geq 1} |z| \nu(y,dz)  < \infty.
\end{equation}
Conditions~\eqref{eq:defLevyMeasure} are part of the definition of a L\'evy measure 
while the conditions~\eqref{eq:PropLevyMeasure} relate to the existence of moments greater than one, see in particular item~4 below.
Note further that these three conditions also hold for both $\nu_0$ and $\nu_1$.
We denote by $\Fc_t^Y$ the filtration generated by $Y$.  Note that $\zeta$, which represents the default time of $X$, is not $\Fc_t^Y$-measurable.  Thus, we introduce an indicator process $D_t := \Ib_{\{\zeta \leq t\}}$ in order to keep track of the event $\{\zeta \leq t \}$.  We denote by $\Fc_t^D$ the filtration generated by $D$.  The filtration of a market observer, then, is $\Fc_t = \Fc_t^Y \vee \Fc_t^D$.  The main features of the class of models described above are as follows:
\begin{enumerate}
\item \textbf{Local volatility: } 
the process $Y$ has a local volatility component: $\sig(y)=\( a_0^2 + \eps a_1^2 \eta(y) \)^{1/2}$.
\item \textbf{Local L\'evy measure: }
jumps in $Y$ of size $dz$ arrive with a state-dependent intensity described by the local L\'evy measure $\nu(y,dz)$.  The L\'evy measure has the decomposition $\nu(y,dz)=\nu_0(dz) + \eps \eta(y) \nu_1(dz)$, where $\nu_0$ and $\nu_1$ are both L\'evy measures.  
Note that both the jump intensity and the jump distribution can change depending on the value of $y$.
\item \textbf{Local default intensity: }
the underlying asset $X$ can default (i.e., for any $t>0$, $\mathbb{P}(X_t=0)>0$) with a state-dependent default intensity of $k(y):=c_0 + \eps c_1 \eta(y) $.
\item \textbf{Martingale: }
the conditions above ensure that $\Eb(X_t)$ is finite for any $t \geq 0$.  
The drift function $\alpha$ is fixed by the L\'evy measure, the volatility and the killing functions,
ensuring that $X$ is a martingale.
\item \textbf{Existence: }
the L\'evy-It\^o SDE \eqref{eq:dY} has a unique strong solution 
(see Theorem~1.19 in~\citet*{oksendal2}).
\end{enumerate}

%
%

\section{Option pricing}
\label{sec:pricing}
Let $V_t$ be the value at time $t$ of a European derivative, expiring at time $T>t$ with payoff $H(X_T)$.  
For convenience, we introduce the function $h:\mathbb{R}\ni y\mapsto H(e^{y})$ with $K:=H(0)$.
Using risk-neutral pricing, $V_t$ is expressed as the conditional expectation of the option payoff
\begin{align}
V_t
        &=      \Eb \[ H(X_T) | \Fc_t \] \\
        &=      \Eb \[ h(Y_T) \Ib_{\{\zeta>T\}} | \Fc_t \] + K \Eb \[ \Ib_{\{\zeta \leq T\}}| \Fc_t \] \\
        &=      \Eb \[ h(Y_T) \Ib_{\{\zeta>T\}} | \Fc_t \] + K - K \Eb \[ \Ib_{\{\zeta > T\}}| \Fc_t \] \\
        &=      \Ib_{\{\zeta>t\}} \Eb \[ h(Y_T) e^{-\int_t^T k(Y_s) ds} | \Fc_t^Y \] 
                        + K - K \Ib_{\{\zeta>t\}} \Eb \[ e^{-\int_t^T k(Y_s) ds} | \Fc_t^Y \] \\
        &=      \Ib_{\{\zeta>t\}} \Eb \[ h(Y_T) e^{-\int_t^T k(Y_s) ds} | Y_t \] 
                        + K - K \Ib_{\{\zeta>t\}} \Eb \[ e^{-\int_t^T k(Y_s) ds} | Y_t \]  , \label{eq:V}
\end{align}
where we have used
\begin{align}
\Eb \[  h(Y_T) \Ib_{\{\zeta>T\}} | \Fc_t \]
        &=      \Ib_{\{\zeta>T\}} \Eb \[  h(Y_T) \Eb[ \Ib_{\{\zeta>T\}} | \Fc_T^Y \vee \Fc_t ] | \Fc_t \] \\
        &=      \Ib_{\{\zeta>T\}} \Eb \[  h(Y_T) e^{-\int_t^T k(Y_s) ds} | \Fc_t \] \\
        &=      \Ib_{\{\zeta>T\}} \Eb \[  h(Y_T) e^{-\int_t^T k(Y_s) ds} | \Fc_t^Y \] .
\end{align}
Using the time-homogeneity of $Y$, it is clear that
$$
\Eb \[ h(Y_T) e^{-\int_t^T k(Y_s) ds} | Y_t = y \]
        =      \Eb_y \[h(Y_{T-t}) e^{-\int_0^{T-t} k(Y_s) ds}\] ,
$$
where the notation $\Eb_y$ means $\Eb[ \cdot | Y_0 = y ]$.  
Thus, to value a European-style derivative, we must compute expectations of the form
\begin{equation} \label{eq:u.eps.def}
u^\eps(t,y) := \Eb_y \left[h(Y_t) e^{-\int_0^{t} k(Y_s) ds}\right].
\end{equation}
We explicitly indicate the dependence of $u^\eps(t,y)$ on the parameter $\eps$, which will play a key role in the regular perturbation analysis below.  
The function~$u^\eps(t,y)$ in~\eqref{eq:u.eps.def} satisfies the Kolmogorov backward equation
\begin{align}
\left(-\d_t + \Ac^\eps \right) u^\eps &= 0,  &
	&\text{with boundary condition }
u^\eps(0,y) = h(y), \label{eq:u.eps.PIDE}
\end{align}
where the infinitesimal generator $\Ac^\eps$ is defined by
\begin{equation}\label{eq:A.limit}
\Ac^\eps f(x)
        =      \lim_{t \to 0^+} \frac{1}{t} \( \Eb_y\left[ f(Y_t) e^{-\int_0^{t} k(Y_s) ds}\right] - f(y) \), 
\qquad\text{whenever the limit exists}.
\end{equation}
If $f \in C_0^2(\Rb)$, then the limit \eqref{eq:A.limit} does exist and the generator $\Ac^\eps$ has the explicit representation
\begin{align}
\Ac^\eps
        &=              \Ac_0 + \eps \eta \Ac_1 , \label{eq:A.eps} \\
\Ac_i
        &=              \frac{1}{2} a_i^2 \( \d^2 - \d \) 
                                + c_i \( \d - 1 \)
                                - \int_\Rb \nu_i(dz) \Big(e^z - 1 - z \Big) \d
                                + \int_\Rb \nu_i(dz) \Big( \theta_z - 1 - z \d \Big) , &
i
        &\in    \{0,1\} , \label{eq:A.i}
\end{align}
where $\d$ (without the subscript $t$) indicates differentiation with respect to $y$ 
and $\theta_z$ is the \emph{shift operator}: $\theta_z f(y) = f(y+z)$.  
We define $\dom{\Ac^\eps}$ as the set of functions $u$ such that the derivatives and integrals appearing in $\Ac^\eps u$ with $\Ac^\eps$ given by \eqref{eq:A.eps}-\eqref{eq:A.i} exist and are finite.
\begin{remark}
\label{rmk:A}
By \citet*{jacod1987limit}, Definition II.8.25, Proposition II.8.26, the operators~$\Ac_0$ and~$\Ac_1$ correspond to infinitesimal generators of L\'evy processes which are exponentially special semimartingales.
\end{remark}
\begin{assumption}
We assume the existence of a unique classical solution to the Cauchy problem \eqref{eq:u.eps.PIDE}.  
A sufficient (but not necessary) condition for its existence is that the payoff function~$h$ and its first two derivatives are bounded (see Theorem 3.2, Chapter 3 in~\citet*{bensoussan1984impulse}).
\end{assumption}
From~\eqref{eq:A.eps}, since the operator $\Ac^\eps$ decomposes into $\Oc(1)$ and $\Oc(\eps)$ terms, 
we seek a solution to the Cauchy problem \eqref{eq:u.eps.PIDE} of the form
\begin{align}
u^\eps
        &=              \sum_{n=0}^\infty \eps^n u_n . \label{eq:u.eps.expand}
\end{align}
Conditions under which this expansion is valid will be given in Theorem \ref{thm.u.eps}.  Inserting the expansion \eqref{eq:u.eps.expand} into the Cauchy problem \eqref{eq:u.eps.PIDE} and collecting terms of like powers of $\eps$ we find
\begin{align}
\Oc(1):&&
( - \d_t + \Ac_0 ) u_0
        &=              0  , &
u_0(0,y)
        &=              h(y) , \label{eq:u0.PIDE} \\
\Oc(\eps^n):&&
( - \d_t + \Ac_0 ) u_n
        &=              - \eta \Ac_1 u_{n-1} , &
u_n(0,y)
        &=              0, &
				&\text{for }n\geq 1 . \label{eq:un.PIDE}
\end{align}
To solve the above Cauchy problems, it will be convenient to introduce the notations
\begin{equation}\label{eq:L2.norm}
\< u , v \> := \int_\Rb \overline{u(y)} v(y)dy 
\qquad\text{and}\qquad
\norm{u}^2 := \< u,u \>.
\end{equation}
Note that the inner product $\< u,v \>$ may be infinite.  
We also introduce $\Ad_i$, the \emph{formal adjoint} of $\Ac_i$ defined via the relation
$\< u , \Ac_i v \> = \< \Ad_i u , v \>$,
for any Schwartz functions~$u$ and~$v$.
Explicitly, $\Ad_i$ is given by
\begin{equation}\label{eq:Ad.i}
\Ad_i = \frac{1}{2} a_i^2 \( \d^2 + \d \) 
                                + c_i \( - \d - 1 \)
                                + \int_\Rb \nu_i(dz) \(e^z - 1 - z \) \d
                                + \int_\Rb \nu_i(dz) ( \theta_{-z} - 1 + z \d), 
\end{equation}
for $i=0,1$, which can be deduced through integrating by parts.  We note the following important relations
\begin{align}
\Ac_0 \psi_\lam
        &=              \phi_\lam \psi_\lam , &
\Ad_0 \overline{\psi_\lam}
        &=              \phi_\lam \overline{\psi_\lam} , &
\Ac_1 \psi_\lam
        &=              \chi_\lam \psi_\lam , &
\Ad_1 \overline{\psi_\lam}
        &=              \chi_\lam \overline{\psi_\lam},
\end{align}
where
\begin{align}
\psi_\lam(y)
        &:=              \frac{1}{\sqrt{2\pi}} e^{i \lam y} , \label{eq:psi}\\
\phi_\lam
        &:=              \frac{1}{2} a_0^2 \( -\lam^2 - i \lam \) 
                                + c_0 (i\lam - 1)
                                - \int_\Rb \nu_0(dz) \Big( e^{z} - 1 - z \Big) i \lam 
                                + \int_\Rb \nu_0(dz) \Big( e^{ i \lam z} - 1 - i \lam z \Big) , \\
\chi_\lam
        &:=              \frac{1}{2} a_1^2 \( -\lam^2 - i \lam \) 
                                + c_1 (i\lam - 1)
                                - \int_\Rb \nu_1(dz) \Big( e^{z} - 1 - z \Big) i \lam 
                                + \int_\Rb \nu_1(dz) \Big( e^{ i \lam z} - 1 - i \lam z \Big) .
\end{align}
Note that for any function $u$ and any complex number $\lam:=\lam_r + i\lam_i \in \Cb$ such that $\< \psi_\lam , u\>$ is finite, we have the \emph{generalized} Fourier representation
\begin{equation}\label{eq:identity}
u(y) = \int_\Rb d\lam_r \< \psi_\lam, u \> \psi_\lam(y).
\end{equation}
However, whenever $u\in L^2(\Rb)$, such a generalized representation is not necessary, 
and the simpler form (with $\lambda\in\Rb$)
$u(y) = \int_\Rb d\lam \< \psi_\lam, u \> \psi_\lam(y)$ suffices.
We are now in a position to find an explicit solution to~\eqref{eq:u0.PIDE}-\eqref{eq:un.PIDE}.
\begin{proposition}
\label{thm:un}
Suppose that $u_0$ satisfies \eqref{eq:u0.PIDE}.
Then the sequence $(u_n)$ defined in~\eqref{eq:un.PIDE} reads
\begin{align}
u_n(t,y)
        &=              \underbrace{ \int_\Rb \cdots \int_\Rb}_{n+1} \( \prod_{k=0}^n d\lam_k \)
                                \( \sum_{k=0}^n \frac{e^{t \phi_{\lam_k}}}{\prod_{j \neq k}^n (\phi_{\lam_k}-\phi_{\lam_j})}\)
                                \( \prod_{k=0}^{n-1} \chi_{\lam_k} \< \psi_{\lam_{k+1}}, \eta \psi_{\lam_k} \> \) \< \psi_{\lam_0},h \> \psi_{\lam_n}(y) , 
                                \label{eq:u.n}
\end{align}
where $\prod_{k=0}^{-1} (\cdots) = 1$ and $\prod_{j \neq k}^0 (\cdots) = 1$ by convention.
\end{proposition}
\begin{proof}
See Appendix \ref{sec:un}.
\end{proof}
\noindent
We have obtained a formal expansion (\eqref{eq:u.eps.expand} and \eqref{eq:u.n}) for $u^\eps$.  
The following theorem provides precise conditions under which the expansion is guaranteed to be valid.
From now on, we  shall denote by~$L^2(\Rb)$ the set of all real functions which are square integrable with respect to the Lebesgue measure.
\begin{theorem}
\label{thm.u.eps}
Suppose $h \in L^2(\Rb) \cap \dom{\Ac^\eps}$.  Suppose further that for any $u \in L^2(\Rb)$
\begin{align}
\int_\Rb d\lam |\phi_\lam|^2 |\<\psi_\lam,u\>|^2 < \infty
        &\qquad\text{implies}\qquad \int_\Rb d\lam |\chi_\lam|^2 |\<\psi_\lam,u\>|^2 < \infty, \label{eq:implies}
\end{align}
and that there exist two real constants $A\geq 0$ and $B\leq 1$ (independent of $(t,y)$) such that
\begin{align}
\eps^2
        &\leq    \inf_{\lam \in \Rb} \frac{A^2 + B^2 |\phi_\lam|^2}{\left\|\eta \right\|^2 \cdot | \chi_\lam |^2}. \label{eq:eps.bound}
\end{align}
Then the option price~$u^\eps(t,y)$ is an analytic function of~$\eps$
and its power series expansion is given by~\eqref{eq:u.eps.expand} 
where the sequence~$\{u_n\}_{n=0}^\infty$ is given by~\eqref{eq:u.n}.
The sequence of partial sums $u^{(N)}(t,y):= \sum_{n=0}^N \eps^n u_n(t,y)$ converges uniformly (with respect to $\eps$) to the exact price $u^\eps(t,y)$.
\end{theorem}
\begin{proof}
See Appendix \ref{sec:proof2}.
The last convergence statement simply follows from the fact that every power series converges uniformly within its radius of convergence.
\end{proof}
\begin{remark}
\label{rmk:u.exact}
We wish to rectify a common misperception.  Under the conditions of Theorem \ref{thm.u.eps}, the series expansion \eqref{eq:u.eps.expand} is the \emph{exact} option price $u^\eps(t,y)$.  It is \emph{not} an asymptotic approximation.
\end{remark}
\begin{remark}[Feynman-Kac transition densities]
Since the diffusion component~$\sigma$ of $Y$ is non-zero, as assumed in Section~\ref{sec:assumptions},
the function $u^\eps(t,y)$ can be written as an integral with respect to a density
\begin{align}
u^\eps(t,y) 
        &:=     \Eb_y \left(h(Y_t) e^{-\int_0^t k(Y_s) ds}\right)
        =       \int_\Rb h(z) p^\eps(t,y,z) dz . \label{eq:fk.density}
\end{align}
The density $p^\eps(t,y,z)$ is called the \emph{Feynman-Kac} (FK) transition density.  
However it is not a probability density since, due to the killing function $k(y)$, it is norm-defecting, i.e.
$\int_\Rb p^\eps(t,x,z) dz \leq 1$.
If one sets the payoff function $h=\del_{z}$, then $u(t,x)$ becomes the FK density $p(t,x,z)$ since
$\int_\Rb \del_z(z') p^\eps(t,x,z') dz' =      p^\eps(t,x,z)$.
Strictly speaking, the Dirac delta $\del_z$ is not in $L^2(\Rb)$, but is a densely defined unbounded linear functional on the Hilbert space $L^2(\Rb)$. 
Its action on functions in $L^2(\Rb)$ is well-defined.
In particular, by making the replacement $\< \psi_{\lam_0}, h \> \to \< \psi_{\lam_0},\del_z \>=\tfrac{1}{\sqrt{2\pi}}e^{-i\lam_0 z}$ in \eqref{eq:u.n}, one obtains $p^\eps(t,y,z)$.
\end{remark}
\begin{remark}[European calls and puts]
The most common European options---calls and puts---have payoffs~$h$ which do not belong to~$L^2(\Rb)$.  
Assuming the expectation~\eqref{eq:u.eps.def} is finite, one can still obtain the price of such an option by integrating the payoff against the FK density $p^\eps$, as in~\eqref{eq:fk.density}.  
However, a more computationally convenient means of obtaining the option price is to use the method of generalized Fourier transforms.  
Note that, even when $h \notin L^2(\Rb)$, the inner product $\< \psi_\lam, h \>$ appearing in~\eqref{eq:u.n} can sometimes be made finite by fixing an imaginary component of $\lam$.  
A European call option, for example, has a payoff $h(y) \equiv (e^y - e^k)^+$, which has a generalized Fourier transform
\begin{align}
\< \psi_\lam ,h \>
                &=              \int_\Rb dy \tfrac{1}{\sqrt{2\pi}}e^{-i \lam y} \left(e^y - e^k\right)^+ 
                =                       \frac{-e^{k-i k \lam}}{\sqrt{2 \pi } \(i \lam + \lam^2\)} ,
\label{eq:hhat}
\end{align}
where $\lam = \lam_r + i \lam_i$ and $\lam_i\in (-\infty, -1)$.
As such, one can still use~\eqref{eq:u.n} to compute call options.
Indeed one simply fixes an imaginary component $\Im[\lam_0] < 1$ 
and integrates with respect to the real part $\Re[\lam_0]$.
\end{remark}

%
%

\section{Implied volatility}
\label{sec:impvol}
For European calls and puts, it is often the implied volatility induced by an option price, rather than the option price itself, that is of primary importance.  
It is therefore fundamental to be able to compute them.  
In this section, we derive an implied volatility expansion for a certain sub-class of the model above.
\begin{assumption}
\label{ass:phi}
In \underline{this section only} we assume $\nu_0 \equiv 0$ and $c_0=0$, which implies
$\phi_\lam \equiv -\frac{a_0^2}{2} \(\lam^2 + i \lam \)$.
\end{assumption}
To begin our implied volatility analysis, we fix a time to maturity $t>0$, an initial value of the underlying $Y_0=y$ and a call option payoff $h(y) = (e^y-e^k)^+$.  
Our goal is to find the implied volatility (defined below) for \emph{this particular option}.  
For ease of notation, throughout this section, we will suppress all dependence on $(t,y,k)$.  
The reader should keep in mind, however, that the implied volatility \emph{does} depend on these variables.  
We begin our analysis by defining the \emph{Black-Scholes price} and \emph{the implied volatility}.
\begin{definition}
The \emph{Black-Scholes Price} $u^{\BS}:\Rb^+ \to \Rb^+$, defined as a function of volatility $\sig$, is given by
\begin{align}
u^{\BS}(\sig)
        &:=             \int_{\Rb} d\lam_r e^{t \phi^{\BS}_\lam(\sig)} \<\psi_\lam, h\> \psi_\lam, &
\phi^{\BS}_\lam(\sig)
        &=              -\frac{\sig^2}{2}(\lam^2 + i \lam) . \label{eq:u.BS}
\end{align}
\end{definition}
\begin{remark}
\label{rmk:u0=uBS}
Note that Equation~\eqref{eq:u.n}, together with Assumption~\ref{ass:phi} imply that $u_0=u^{\BS}(a_0)$.
\end{remark}
\begin{remark}
Usually, the Black-Scholes price is written as
$u^{\BS}(\sig) = \int_\Rb p^{\BS}(t,y,z) h(z) dy$;
Expression~\eqref{eq:u.BS} is simply its Fourier representation.  
Here 
$p^{\BS}(t,y,z) \equiv \frac{1}{\sigma\sqrt{2 \pi t}} \exp \( \frac{(z - (y - \sig^2 t/2) )^2}{2 \sig^2 t}\)$
is the transition density of a Brownian motion with drift $-\sig^2/2$ and volatility $\sig$.  
We use the Fourier representation of $u^{\BS}$ as it will be more convenient for the analysis that follows.
\end{remark}
\begin{definition}\label{eq:imp.vol.def}
For an option price $u^\eps$, the implied volatility is defined implicitly as the unique number 
$\sig^\eps \in \Rb^+$ such that
$u^{\BS}(\sig^\eps) =u^\eps$. 
\end{definition}
\begin{remark}
\label{rmk:exist}
For any $t>0$, the existence and uniqueness of the implied volatility~$\sig^\eps$ follows from the general arbitrage bounds for call prices and the monotonicity of $u^{\BS}$ (\cite{fpss}, Section 2.1, Remark~(i)).
\end{remark}
\begin{remark}
\label{rmk:analytic}
For any $\sig_0>0$ and $\sig_0 + \del>0$, the function $u^{\BS}(\sig_0 + \del)$ is given by its Taylor series:
\begin{align}
u^{\BS}(\sig_0 + \del)
        &=      \sum_{n=0}^\infty \frac{\del^n}{n!} \d_\sig^n u^{\BS}(\sig_0), &
\d_\sig^n u^{\BS}(\sig_0)
        &=              \int_\Rb d\lam_r \( \d_\sig^n e^{t \phi^{\BS}_\lam(\sig_0)} \) \< \psi_\lam, h\> \psi_\lam . \label{eq:d.sigma}
\end{align}
Observe also that, by monotonicity of $u^{\BS}$ we have $\d_\sig u^{\BS}(\sig)> 0$ for all $\sig>0$.  Therefore, $u^{\BS}$ is an invertible analytic function, as the following theorem shows.
\end{remark}
\begin{theorem}[Lagrange Inversion Theorem]
\label{thm:lagrange}
Suppose $u$ is defined as a function of $\sig$ through the equation $u^{\BS}(\sig)=u$, where $u^{\BS}$ is analytic at a point $\sig_0$ and $\d_\sig u^{\BS}(\sig_0) \neq 0$.  Then it is possible to solve for $\sig$ on a neighborhood of $u^{\BS}(\sig_0)$ where $[u^{\BS}]^{-1}$ is analytic:
\begin{align}
\sig
        &=      \[u^{\BS}\]^{-1}(u)
        =               \sig_0 + \sum_{n=1}^\infty \frac{b_n}{n!} \(u - u^{\BS}(\sig_0)\)^n ,  &
b_n
        &=      \lim_{\sig \to \sig_0} \d_\sig^{n-1} \( \frac{ \sig - \sig_0}{u^{\BS}(\sig)-u^{\BS}(\sig_0)} \)^n . \label{eq:inverse}
\end{align}
\end{theorem}
\begin{proof}
See \citet*{as}, Equation 3.6.6.
\end{proof}
\noindent
Theorem \ref{thm:lagrange} shows that, for every fixed $\sig_0>0$, the exists some radius of convergence $R>0$ such that $|u - u^{\BS}(\sig_0)|<R$ implies $\sig$, defined through $u^{\BS}(\sig)=u$, is given by \eqref{eq:inverse}.  The radius of convergence $R$ depends on the coefficients $b_n$, which, in general, are quite difficult to compute.  Note, from the expression for $b_n$, the radius of convergence $R$ depends on $(t,y,k)$ through the function $u^{\BS}$.

Recall that Theorem~\ref{thm.u.eps} shows that $u^\eps$ is an analytic function of $\eps$.  
Since the composition of two analytic functions is also analytic (\citet*{brown1996complex}, section 24, p. 74), Theorem \ref{thm:lagrange} implies that the implied volatility $\sig^\eps = [u^{\BS}]^{-1}(u^\eps)$ is an analytic function of $\eps$, and therefore has a power series expansion.  
We write this expansion as
\begin{align}
\sig^\eps
        &=              \sig_0 + \del^\eps , &
\del^\eps
        &=              \sum_{k=1}^\infty \eps^k \sig_k  . \label{eq:sigma.expand}
\end{align}
Taylor expanding $u^{\BS}$ about the point $\sig_0$ we have
\begin{align}
u^{\BS}(\sig^\eps)
        &=      u^{\BS}(\sig_0 + \del^\eps) \\
        &=              \sum_{n=0}^\infty \frac{1}{n!}(\del^\eps \d_\sig )^n u^{\BS}(\sig_0) \\
        &=              u^{\BS}(\sig_0) +
                                                \sum_{n=1}^\infty \frac{1}{n!} \( \sum_{k=1}^\infty \eps^k \sig_k \)^n \d_\sig^n u^{\BS}(\sig_0) \\
        &=              u^{\BS}(\sig_0) + 
                                                \sum_{n=1}^\infty \frac{1}{n!}  
                                                \[ \sum_{k=1}^\infty \( \sum_{j_1+\cdots+j_n=k} \prod_{i=1}^n \sig_{j_i} \) \eps^k \] \d_\sig^n u^{\BS}(\sig_0) \\
        &=              u^{\BS}(\sig_0) +
                                                \sum_{k=1}^\infty \eps^k 
                                                \[ \sum_{n=1}^\infty \frac{1}{n!} \( \sum_{j_1+\cdots+j_n=k} \prod_{i=1}^n \sig_{j_i} \) \d_\sig^n \] u^{\BS}(\sig_0) \\
        &=      u^{\BS}(\sig_0) +
                                                \sum_{k=1}^\infty \eps^k 
                                                \[ \sig_k \d_\sig + \sum_{n=2}^\infty \frac{1}{n!}\( \sum_{j_1+\cdots+j_n=k} \prod_{i=1}^n \sig_{j_i} \)  \d_\sig^n \]
                                                u^{\BS}(\sig_0) .                        \label{eq:u.bs.expand}
\end{align}
Now, we insert the expansions \eqref{eq:u.eps.expand} and \eqref{eq:u.bs.expand} 
into the definition~\ref{eq:imp.vol.def} and collect terms of like order in $\eps$:
\begin{align}
\Oc(1):&&
u_0
        &=              u^{\BS}(\sig_0) , \\
\Oc(\eps^k):&&
u_k
        &=              \sig_k \d_\sig u^{\BS}(\sig_0)
                                + \sum_{n=2}^\infty \frac{1}{n!}\( \sum_{j_1+\cdots+j_n=k} \prod_{i=1}^n \sig_{j_i} \)  \d_\sig^n  u^{\BS}(\sig_0) , &
k
        &\geq   1 .
\end{align}
Solving the above equations for the sequence $(\sig_k)_{k\geq 0}$ we find
\begin{equation}\label{eq:sig.k}
\left.
\begin{array}{rlll}
\Oc(1): 
& \sig_0 = a_0, & \\
\Oc(\eps^k):
 & \sig_k = \displaystyle \frac{1}{\d_\sig u^{\BS}(\sig_0)}
\( u_k - \sum_{n=2}^\infty \frac{1}{n!}\( \sum_{j_1+\cdots+j_n=k} \prod_{i=1}^n \sig_{j_i} \)  \d_\sig^n  u^{\BS}(\sig_0)\),
 & k \geq 1, 
\end{array}
\right.
\end{equation}
where we have used Remark \ref{rmk:u0=uBS} to deduce that $\sig_0 = a_0$.  
\begin{remark}
The sequence~$(\sig_k)_{k\geq 1}$ can be determined recursively since~\eqref{eq:sig.k} only depends on $(\sig_j)_{j \leq k-1}$.
\end{remark}
\begin{remark}
Note that $\d_\sig^n u^{\BS}(\sig)$ can be easily computed using \eqref{eq:d.sigma}.
\end{remark}
\noindent
Explicitly, up to $\Oc(\eps^4)$ we have
\begin{align}
\Oc(\eps):&&
\sig_1
        &=      \frac{u_1}{\d_\sig u_0}, \label{eq:sig.1} \\
\Oc(\eps^2):&&
\sig_2
        &=      \frac{u_2 - \tfrac{1}{2} \sig_1^2 \d_\sig^2 u_0}{\d_\sig u_0}, \label{eq:sig.2} \\
\Oc(\eps^3):&&
\sig_3
        &=      \frac{u_3 - (\sig_2 \sig_1 \d_\sig^2 + \tfrac{1}{3!}\sig_1^3 \d_\sig^3) u_0}{\d_\sig u_0}, \\
\Oc(\eps^4):&&
\sig_4
        &=  \frac{u_4 - (\sig_3 \sig_1 \d_\sig^2 + \tfrac{1}{2} \sig_2^2 \d_\sig^2 
                                                                        + \tfrac{1}{2} \sig_2 \sig_1^2 \d_\sig^3 + \tfrac{1}{24} \sig_1^4 \d_\sig^4) u_0}{\d_\sig u_0}.
\end{align}
We summarize our implied volatility result in the following theorem:
\begin{theorem}[Implied volatility]
\label{thm:imp.vol}
Let $R(t,y,k,a_0)$ denote the radius of convergence of the infinite series~\eqref{eq:inverse} 
with~$\sig_0 = a_0$.
Assume further that~$\eps$ satisfies~\eqref{eq:eps.bound} and that
$|u^\eps - u^{\BS}(a_0)| = \left|\sum_{n=1}^\infty \eps^n u_n\right| < R(t,y,k,a_0)$.
Then the implied volatility $\sig^\eps$ (Definition~\ref{eq:imp.vol.def}), 
is characterised by~\eqref{eq:sigma.expand}, 
with $(\sig_k)_{k\geq 0}^\infty$ given by~\eqref{eq:sig.k}.
\end{theorem}
\begin{remark}
\label{rmk:sig.exact}
We emphasize that, within the radius of convergence, the implied volatility expansion is \underline{exact}.  
It is not an asymptotic approximation.  
That is, for fixed $(t,y,k)\in \Rb_+\times\Rb\times\Rb$, the sequence of partial sums $\sig^{(N)}:= \sum_{n=0}^N \eps^n \sig_n$ converges to the exact implied volatility $\sig^\eps$ (i.e., we have pointwise convergence).  
The convergence is uniform with respect to $\eps$ (since every power series converges uniformly within its radius of convergence).  
Furthermore, while the accuracy of the implied volatility expansion \eqref{eq:sigma.expand} is limited by the number of terms one wishes to compute, we will show through a numerical example in Section \ref{sec:example} that very few terms are actually required to achieve an accurate approximation of implied volatility.
\end{remark}
\begin{remark}
As written, the expansion~\eqref{eq:sigma.expand} with $\sig_k$ given by \eqref{eq:sig.k} is not very convenient to compute.  Indeed, $\sig_k$ in \eqref{eq:sig.k} requires computing $u_k$ which, in the general case, requires a $(k+1)$-fold numerical integral.  Thus, the results of this section are primarily of theoretical interest.  In Section \ref{sec:example} however, we will show that, in a CEV-like setting, an approximation of the implied volatility can be computed in closed form.
\end{remark}

%
%

\section{Example: CEV-like L\'evy-type process}
\label{sec:example}
The constant elasticity of variance (CEV) model of \citet*{CoxCEV} improves upon the Black-Scholes model by allowing the volatility to depend on the present level of the underlying through a local volatility function of the form $\sig(y) = a e^{\beta y}$ (recall, $y=\log x$).  This model has enjoyed wide success because (i) it admits closed-form solutions for European option prices and (ii) when $\beta < 0$, the local volatility function \emph{increases} as $y \to - \infty$, which is consistent with the leverage effect and results in a negative implied volatility skew.  Still, the CEV model has some shortcomings.  First, the volatility function $\sig(y)$ drops to zero as $y$ tends to infinity.  Second, the model does not allow the underlying to experience jumps.
\par
We can retain some CEV-like features, while overcoming both of the above mentioned shortcomings by choosing $\eta(y) \equiv e_\beta(y):=e^{\beta y}$ in our framework.  In this setting, the volatility function, killing function, and L\'evy measure become 
\begin{align}
\sig(y)
        &=      (a_0 + \eps a_1^2 e^{\beta y})^{1/2} , &
k(y)
        &=      c_0 + \eps c_1^2 e^{\beta y} , &
\nu(y,dz)
        &=      \nu_0(y,dz) + \eps e^{\beta y} \nu_1(y,dz) .
\end{align}
To maintain consistency with the leverage effect, and to simplify the discussion, we shall assume that $\beta \leq 0$.
\begin{remark}
\label{rmk:modify}
Note that, since $e_\beta \notin \Sc$, the CEV-like model described above does not belong to the class of models described in Section \ref{sec:assumptions}.  
Nevertheless, one can always fix some $\underline{y}<Y_0$ and modify the function~$\eta$ 
so that $\eta\equiv e_\beta$ on the open interval $(\underline{y}, \infty)$ and so that it decays smoothly to zero 
on~$(-\infty, \underline{y}]$.
In this case, one should verify that the perturbing parameter $\eps$ is small enough to satisfy~\eqref{eq:eps.bound} with the modified function~$\eta$.
Throughout this section we will continue to perform computations with $\eta \equiv e_\beta$.  
We will check the validity of this simplification by testing our results by Monte Carlo simulation.
One could in principle make this adjustment more precise:
define $\tau_\eps:=\inf\{t\geq 0: Y_t<-\eps^{-1}\}$.
Then, if for any $t>0$, the quantity $\log\mathbb{P}(\tau_\eps<t)$ decays at least as fast as $-\eps^{-1}$, 
then we can modify the coefficients of the process such that the new process has similar tails (on an exponentially decreasing scale).
Such an argument can be found for instance in (\citet*{DFJV1}, Remark 2.11)
\end{remark}
\begin{remark}
\label{rmk:zeta}
When $\eta \equiv e_\beta$, the process $Y$ may reach $-\infty$ in finite time (equivalently, 
the origin is an attainable boundary for~$X$).
To account for this, we modify the default time $\zeta$ to be $\zeta := \zeta_0 \wedge \zeta_1$, where
$\zeta_{0} := \inf \{ t \geq 0 : Y_t = - \infty \}$
and
$\zeta_{1} := \inf \{ t \geq 0 : \int_0^t k(Y_s) ds \geq \Ec \}$.
This construction (see for example Section 1.1 in~\citet*{linetskybook}), corresponds to specifying $-\infty$ (resp. $0$)
as an absorbing boundary for $Y$ (resp. $X$).
\end{remark}
The CEV-like model enjoys the follow features:
\begin{itemize}
\item The local volatility function $\sig(y) \equiv \( a_0^2 + \eps a_1^2 e^{\beta y} \)^{1/2}$ behaves asymptotically 
like the CEV model $\sig(y) \sim \sqrt{\eps} a_1 e^{\beta y/2}$ as $y$ decreases to $-\infty$, reflecting the fact that volatility tends to increase as the asset price drops (the leverage effect).  
However, $\lim_{y\nearrow +\infty}\sigma(y)=a_0$, which is in \emph{contrast} to the CEV model, 
in which the local volatility function drops close to zero as $y$ tends to infinity.
\item Jumps of size $dz$ arrive with a state-dependent intensity of $\nu(y,dz) \equiv \nu_0(dz) + \eps e^{\beta y} \nu_1(dz)$.  
The local L\'evy measure behaves like $\nu(y,dz) \sim \eps e^{\beta y} \nu_1(dz)$ as $y \searrow -\infty$ 
and asymptotically like $\nu(y,dz) \sim \nu_0(dz)$ as $y \nearrow +\infty$.  
Thus, both the jump intensity and jump distribution can change drastically depending on the value of $y$ 
and the choice of L\'evy measures $\nu_0(dz)$ and $\nu_1(dz)$.
\item A default (i.e., of jump to zero of the asset price $X$) arrives with a state-dependent 
intensity~$k(y) \equiv \(c_0 + \eps c_1 e^{\beta y} \)$.  
The local killing function $k$ behaves asymptotically like $\eps c_1 e^{\beta y}$ as $y \searrow -\infty$, reflecting the fact that a default is more likely to occur as the asset price drops.
However, $\lim_{y\nearrow+\infty} k(y) = c_0$, which is a form first suggested by \citet*{JDCEV}.
\end{itemize}
To value an option, we must find an expression for $u_n$, given by \eqref{eq:u.n}, when $\eta = e_\beta$.  For any complex $\lam \in \Cb$ and analytic function $f$, \citet*{diracdelta} shows that
$\frac{1}{2\pi}\int_\Rb e^{i\lam x} dx =\del(\lam)$ 
and 
$\int_\Rb \del(\lam - \mu) f(\mu) d\mu = \lam$.
Thus, with $\psi_\lam$ given by \eqref{eq:psi}, we have
\begin{align}
\<\psi_\mu , e_\beta \psi_\lam \>
        &=              \del(\lam-\mu-i\beta) . \label{eq:delta}
\end{align}
Inserting \eqref{eq:delta} into \eqref{eq:u.n}, we see that the $(n+1)$-fold integral collapses into a single integral
\begin{align}
u_n
        &=              \int_\Rb d\lam \( \sum_{k=0}^n \frac{e^{t \phi_{\lam-ik\beta}}}
                                {\prod_{j\neq k}^n (\phi_{\lam-ik\beta}-\phi_{\lam-ij\beta})}\)
                                \( \prod_{k=0}^{n-1} \chi_{\lam-ik\beta}\) \<\psi_\lam, h\> \psi_{\lam-in\beta} \\
        &=              e_{n\beta} \int_\Rb d\lam \( \sum_{k=0}^n \frac{e^{t \phi_{\lam-ik\beta}}}
                                {\prod_{j\neq k}^n (\phi_{\lam-ik\beta}-\phi_{\lam-ij\beta})}\)
                                \( \prod_{k=0}^{n-1} \chi_{\lam-ik\beta}\) \<\psi_\lam, h\> \psi_{\lam} .
                                \label{eq:u.n.2}
\end{align}
\begin{remark}
\label{rmk:numerics}
Although we have written the option price as an infinite series \eqref{eq:u.eps.expand}, 
from a practical standpoint, one may only compute $u^\eps \approx u^{(N)} := \sum_{n=0}^N \eps^n u_n$ 
for some finite $N$.  
For any such $N$ we may pass the sum $\sum_{n=0}^N$ through the integral appearing in \eqref{eq:u.n.2}.  
Thus, for the purposes of computation, the best way express the approximate option price is
\begin{align}
u^\eps
        &\approx                u^{(N)}
        =                       \int_\Rb d\lam \<\psi_\lam, h\> \psi_{\lam} \sum_{n=0}^N 
                                \eps^n e_{n\beta} \( \sum_{k=0}^n \frac{e^{t \phi_{\lam-ik\beta}}}
                                {\prod_{j\neq k}^n (\phi_{\lam-ik\beta}-\phi_{\lam-ij\beta})}\)
                                \( \prod_{k=0}^{n-1} \chi_{\lam-ik\beta}\) . \label{eq:price.CEV}
\end{align}
Note that, to obtain the approximate value of $u^\eps$, \emph{only a single integration is required}.  
This makes the pricing formula \eqref{eq:price.CEV} as efficient as other models in which option prices are expressed as a Fourier-type integral (e.g. exponential L\'evy processes, Heston model, etc.).
\end{remark}
\begin{remark}
\label{rmk:choice}
The choice $\eta\equiv e_\beta$ is convenient since the Fourier transform of an exponential yields a Dirac Delta function (see~\eqref{eq:delta}), 
which results in the $(n+1)$-fold integral for $u_n$ collapsing to a one-dimensional integral.  
However, $\eta=e_\beta$ is not the only convenient choice for which this occurs.  Observe that
\begin{align}
\left. \begin{aligned}
\<\psi_\lam,\sin(\om \cdot) \>
        &= i \sqrt{\frac{\pi }{2}} \Big( - \del(\lam - \om) + \del(\lam + \om) \Big), \\
\<\psi_\lam,\cos(\om \cdot) \>
        &= \sqrt{\frac{\pi }{2}} \Big( \del(\lam - \om) + \del(\lam + \om) \Big), \\
\<\psi_\lam,(\cdot)^n \>
        &= i^n \sqrt{2\pi} \del^{(n)}(\lam) , 
\end{aligned} \right\} \label{eq:y.n}
\end{align}
where $\del^{(n)}$ is the $n$th derivative of a Delta function.  
In particular, any smooth function $\eta$ can locally be approximated by a truncated power series $\eta(y) \approx \sum_{i=0}^n \frac{1}{n!} \d^n \eta(y_0) (y-y_0)^i$.  Similarly, any periodic function can be approximated by a truncated Fourier series $\eta \approx \sum_{i=0}^n \( a_i \sin( \om_i y) + b_i \cos( \om_i y) \)$.  
Thus, equation \eqref{eq:y.n} provides a way to include arbitrary local dependence.
\end{remark}

%

\subsection{Implied volatility asymptotics for CEV-like models}
\label{sec:impvol.2}
While the implied volatility expansion of Section \ref{sec:impvol} is of considerable theoretical interest, it is not  computationally efficient to use equations \eqref{eq:sigma.expand} and \eqref{eq:sig.k}.  Indeed, computing the value of each $u_i$ in \eqref{eq:sig.k} requires a Fourier integration, which must be done numerically.  However, as we will show, if we restrict our analysis to CEV-like models, the leading order terms for implied volatility can be computed approximately in terms of simple functions, which require no numerical integration.
\begin{assumption}
\label{ass:phi.chi}
To simplify the analysis below, we assume that $\nu_0 \equiv \nu_1 \equiv 0$ and $c_0=c_1=0$, 
(i.e., $Y$ is an It\^o diffusion without killing).  Under this assumption,
$\phi_\lam \equiv -\frac{a_0^2}{2}\(\lam^2 + i \lam \) $
and
$\chi_\lam \equiv -\frac{a_1^2}{2} \(\lam^2 + i \lam \) $.
We emphasise that the assumption on $\nu_1$ is for computational convenience only.  At the end of this section, in Remark \ref{rmk:levy}, we show how to relax this assumption.
\end{assumption}
\noindent
The key to the computations that follow will be to show that $u_1$ and $u_2$ can be approximated by a differential operator acting on $u_0=u^{\BS}(a_0)$.  To this end, using \eqref{eq:u.n.2} we observe that, for any $M\geq 1$, we have
\begin{align}
u_1(t,y)
        &=              e^{\beta y} \int_\Rb d\lambda
                                 \( \frac{e^{t\phi_\lam}}{\phi_\lam-\phi_{\lam-i\beta}} 
                                        + \frac{e^{t \phi_{\lam-i\beta } } }{\phi_{\lam-i\beta}-\phi_\lam } \)
                                \chi_\lam \<\psi_\lam, h\> \psi_{\lam}(y) \\
        &=              e^{\beta y} \int_\Rb d\lambda
                                 \( \frac{1}{\phi_\lam-\phi_{\lam-i\beta}} 
                                        + \frac{e^{t \phi_{\lam-i\beta } - t \phi_\lam} }{\phi_{\lam-i\beta}-\phi_\lam } \)
                                \chi_\lam  e^{t\phi_\lam} \<\psi_\lam, h\> \psi_{\lam}(y) \\
        &=              e^{\beta y} \int_\Rb d\lambda
                                 \( \sum_{n=1}^\infty \frac{t^n}{n!} \( \phi_{\lam-i\beta}-\phi_\lam \)^{n-1} \)
                                \chi_\lam  e^{t\phi_\lam} \<\psi_\lam, h\> \psi_{\lam}(y) \\
        &\approx
                                e^{\beta y} \int_\Rb d\lambda
                                 \( \sum_{n=1}^M \frac{t^n}{n!} \( \phi_{\lam-i\beta}-\phi_\lam \)^{n-1} \)
                                \chi_\lam  e^{t\phi_\lam} \<\psi_\lam, h\> \psi_{\lam}(y) \\
        &=              e^{\beta y} 
                                \sum_{n=1}^M \frac{t^n}{n!} \( \phi_{-i\d-i\beta}-\phi_{-i\d} \)^{n-1} \chi_{-i\d}
                                \int_\Rb d\lambda e^{t\phi_\lam} \<\psi_\lam, h\> \psi_{\lam}(y) \\
        &=              e^{\beta y} 
                                \sum_{n=1}^M \frac{t^n}{n!} \( \phi_{-i\d-i\beta}-\phi_{-i\d} \)^{n-1} \chi_{-i\d}u_0(t,y) 
        =:              u_1^{(M)}(t,y) , \label{eq:u1.M}
\end{align}
We used here the fact that $p(\lam)\psi_\lam=p(-i\d)\psi_\lam$ for any polynomial function $p$.
Similarly,  for $u_2$, we find
\begin{align}
u_2(t,y)
        &=              e^{2 \beta y} \int_\Rb d\lambda
                                \bigg(
                                        \frac{e^{t\phi_\lam}}
                                        { (\phi_\lam-\phi_{\lam-i\beta})( \phi_\lam-\phi_{\lam-2i\beta})}
                                        + \frac{e^{t \phi_{\lam-i\beta}}}
                                        {(\phi_{\lam-i\beta}-\phi_\lam)(\phi_{\lam-i\beta}-\phi_{\lam-2i\beta}) } \\ & \qquad
                                        + \frac{e^{t \phi_{\lam-2i\beta}}}
                                        {(\phi_{\lam-2i\beta}-\phi_\lam)(\phi_{\lam-2i\beta}-\phi_{\lam-i\beta}) }
                                \bigg)
                                \chi_{\lam-i\beta} \chi_\lam \<\psi_\lam, h\> \psi_{\lam}(y) \\
        &=              e^{2 \beta y} \int_\Rb d\lambda
                                \bigg(
                                        \frac{1}
                                        {(\phi_\lam-\phi_{\lam-i\beta})( \phi_\lam-\phi_{\lam-2i\beta})}
                                        + \frac{e^{t \phi_{\lam-i\beta} - t\phi_\lam }}
                                        {(\phi_{\lam-i\beta}-\phi_\lam)(\phi_{\lam-i\beta}-\phi_{\lam-2i\beta}) } \\ & \qquad
                                        + \frac{e^{t \phi_{\lam-2i\beta} - t\phi_\lam}}
                                        {(\phi_{\lam-2i\beta}-\phi_\lam)(\phi_{\lam-2i\beta}-\phi_{\lam-i\beta}) }
                                \bigg)
                                \chi_{\lam-i\beta} \chi_\lam e^{t\phi_\lam} \<\psi_\lam, h\> \psi_{\lam}(y) \\
        &=              e^{2 \beta y} \int_\Rb d\lambda
                                \bigg(
                                        \frac{1}{\phi_{\lam-i\beta}-\phi_{\lam-2i\beta} }
                                        \sum_{n=1}^\infty \frac{t^n}{n!} \( \phi_{\lam-i\beta}-\phi_\lam\)^{n-1} \\ & \qquad
                                        + \frac{1}{\phi_{\lam-2i\beta}-\phi_{\lam-i\beta} }
                                        \sum_{n=1}^\infty \frac{t^n}{n!} \( \phi_{\lam-2i\beta}-\phi_\lam\)^{n-1}
                                \bigg)
                                \chi_{\lam-i\beta} \chi_\lam e^{t\phi_\lam} \<\psi_\lam, h\> \psi_{\lam}(y) \\
        &=              e^{2 \beta y} \int_\Rb d\lambda
                                \bigg(
                                        \sum_{n=2}^\infty \frac{t^n}{n!} \sum_{k=1}^{n-1} \binom{n-1}{k}
                                        \frac{ \( \phi_{\lam-i\beta} \)^k - \( \phi_{\lam-2i\beta} \)^k }
                                        {\phi_{\lam-i\beta}-\phi_{\lam-2i\beta} }
                                        \( \phi_\lam \)^{n-1-k} 
                                \bigg) 
                                \chi_{\lam-i\beta} \chi_\lam e^{t\phi_\lam} \<\psi_\lam, h\> \psi_{\lam}(y) \\
        &=              e^{2 \beta y} \int_\Rb d\lambda
                                \bigg(
                                        \sum_{n=2}^\infty \frac{t^n}{n!} \sum_{k=1}^{n-1} \binom{n-1}{k}
                                        \( - \phi_\lam \)^{n-1-k} 
                                        \sum_{m=0}^{k-1} \( \phi_{\lam-i\beta} \)^m \( \phi_{\lam-2i\beta} \)^{k-1-m}
                                \bigg) \\ & \qquad
                                \chi_{\lam-i\beta} \chi_\lam e^{t\phi_\lam} \<\psi_\lam, h\> \psi_{\lam}(y) \\
        &\approx
                                e^{2 \beta y} 
                                \sum_{n=2}^M \frac{t^n}{n!} \sum_{k=1}^{n-1} \binom{n-1}{k}
                                \( - \phi_{-i\d} \)^{n-1-k} 
                                \sum_{m=0}^{k-1} \( \phi_{-i\d-i\beta} \)^m \( \phi_{-i\d-2i\beta} \)^{k-1-m}
                                \chi_{-i\d-i\beta} \chi_{-i\d} u_0(t,y) \\              
        &=: u_2^{(M)}(t,y) , \label{eq:u2.M}
\end{align}
where we use $b^k-c^k = (b-c)\sum_{n=0}^{k-1}b^n c^{k-1-n}$.  
Define $\sig_1^{(M)}$ and $\sig_2^{(M)}$ as the $M$th order approximation of $\sig_1$ and $\sig_2$
(obtained by replacing $u_1$ and $u_2$ in ~\eqref{eq:sig.1} and~\eqref{eq:sig.2} by $u_1^{(M)}$ and $u_2^{(M)}$):
\begin{equation}\label{eq:sig.2.M}
\left.
\begin{array}{rll}
\Oc(\eps): & 
\sig_1^{(M)}  :=      \displaystyle \frac{u_1^{(M)}}{\d_\sig u_0}, \\
\Oc(\eps^2): &
\sig_2^{(M)} :=      \displaystyle \frac{u_2^{(M)} - \tfrac{1}{2!} (\sig_1^{(M)})^2 \d_\sig^2 u_0}{\d_\sig u_0}, 
\end{array}
\right.
\end{equation}
Since $\chi_{-i\d}\equiv\frac{1}{2}a_1^2(\d^2-\d)$, the functions $u_1^{(M)}$ and $u_2^{(M)}$ are of the form
\begin{equation}
u_i^{(M)}
        =      \sum_{n = 0}^{M} b_{i,n} \d^n (\d^2 - \d) u_0, \label{eq:u.i.M}
\end{equation}
for $i \in \{1,2\}$, where $(b_{i,n})$ are coefficients which can be computed 
by expanding the terms in~\eqref{eq:u1.M} and~\eqref{eq:u2.M}.  
Next, using the Black-Scholes formula for European call options we compute
(recall that $u_0\equiv u^{\BS}$)
\begin{align}
\left.\d_\sig u_0\right|_{\sigma=a_0}
        &= t a_0 ( \d^2 - \d ) u_0 &
(\d^2 - \d ) u_0
        &=      \frac{1}{a_0 \sqrt{t}} \exp \( y - \frac{d_+^2}{2} \) , &
d_+
        &=      \frac{1}{a_0 \sqrt{t}} \( y - k + \frac{a_0^2 t}{2}\) \label{eq:BS}
\end{align}
Inserting~\eqref{eq:u.i.M} and~\eqref{eq:BS} into~\eqref{eq:sig.2.M}, we obtain
\begin{align}
\Oc(\eps):&&
\sig_1^{(M)}
        &=      \sum_{n\geq 0}^{M} \( \frac{b_{1,n} \d^n  \exp\( y - \frac{d_+^2}{2} \)}
                        {t a_0 \exp\( y - \frac{d_+^2}{2} \)} \) , \\
\Oc(\eps^2):&&
\sig_2^{(M)}
        &=      \sum_{n\geq 0}^{M} \( \frac{ b_{2,n} \d^n  \exp\( y - \frac{d_+^2}{2} \) }
                        {t a_0 \exp\( y - \frac{d_+^2}{2} \)} \) -
                        \frac{1}{2}\( \sig_1^{(M)} \)^2 \( \frac{(k-y)^2}{t a_0^3}-\frac{t a_0 }{4} \) .
\end{align}
The above expressions, while perhaps involving many terms, can be easily computed using a computer algebra system such as Mathematica.  Once computed explicitly, the above expressions are simple functions of $(t,y,k)$, which require no integration.  Thus, the approximate implied volatility
\begin{align}
\sig^{(2,M)}
        &:=     \sig_0 + \eps \sig_1^{(M)} + \eps^2 \sig_2^{(M)} ,  \label{eq:sig.2.M.b}
\end{align}
can be computed extremely quickly.  We provide Mathematica code for computing $\sig^{(2,M)}$ in Appendix \ref{sec:mathematica}.
\begin{remark}
\label{rmk:levy}
The results of this section can be further extended by relaxing the assumption on $\chi_\lam$.  Consider the case where the L\'evy measure $\nu_1$ is non-zero.  Then $\chi_\lam$ is of the form
\begin{align}
\chi_\lam 
        &=      \frac{1}{2}a_0^2 \( -\lam^2 - i \lam \) + \int_\Rb \nu_1(dz) (e^{i\lam z} - 1 - i \lam z)
                        -  i \lam \int_\Rb \nu_1(dz) (e^{z} - 1 - z) \\
        &=      \frac{1}{2}a_0^2 \( (i \lam)^2 - i \lam \) + \sum_{n=2}^\infty I_n ( (i \lam)^n - i \lam ) , 
        \qquad
\text{where} \qquad I_n
        :=      \int_\Rb \nu_1(dz) z^n .
\end{align}
In this case, we can approximate the operator $\chi_{-i\d}$ by truncating the infinite sum at some finite $q\in\mathbb{N}$:
\begin{align}
\chi_{-i\d}^{(q)}
        &:= \frac{1}{2}a_0^2 \( \d^2 - \d \) + \sum_{n=2}^q I_n ( \d^n - \d ) \\
        &=      \frac{1}{2}a_0^2 \( \d^2 - \d \) + \sum_{n=2}^q I_n \sum_{k=2}^n ( \d^k - \d^{k-1} ) \\
        &=      \frac{1}{2}a_0^2 \( \d^2 - \d \) + \sum_{n=2}^q I_n \sum_{k=2}^n \d^{k-2}( \d^2 - \d ) ,
\end{align}
This truncation implies that $u_i^{(M)}$ remains of the form \eqref{eq:u.i.M}, 
which allows for explicit computation of~$\sig^{(2,M)}$.
\end{remark}


\subsection{Numerical Results}
Because $\eta = e_\beta$ does not satisfy the requirement $\eta \in \Sc$, it is important to test the validity of the pricing formula \eqref{eq:price.CEV}.  Below, we provide numerical tests to support this formula.  First, we examine convergence of the FK density.  Next, we compare the implied volatility surface induced by option pricing approximation \eqref{eq:price.CEV} to the implied volatility surface generated by a Monte Carlo simulation.  Then, we examine the implied volatility expansion of Section \ref{sec:impvol}.  We also illustrate the empirical relevance of this model by calibrating a particular CEV-like model with Gaussian jumps to the implied volatility surface of S{\&}P500 index options.  Finally, we examine the implied volatility approximation of Section \ref{sec:impvol.2}.

\subsubsection{Convergence of the approximate FK density}
In order to examine convergence of the FK density $p^\eps(t,y,z)$, we define the $\Oc(\eps^N)$ approximation of the FK density $p^{(N)}(t,y,z)$, given by setting $h = \del_z$ in \eqref{eq:price.CEV}
\begin{align}
p^{(N)}(t,y,z)
        =                       \int_\Rb d\lam \<\psi_\lam, \del_z \> \psi_{\lam}(y) \sum_{n=0}^N 
                                \eps^n e_{n\beta} \( \sum_{k=0}^n \frac{e^{t \phi_{\lam-ik\beta}}}
                                {\prod_{j\neq k}^n (\phi_{\lam-ik\beta}-\phi_{\lam-ij\beta})}\)
                                \( \prod_{k=0}^{n-1} \chi_{\lam-ik\beta}\) . \label{eq:price.CEV.p}
\end{align}
In Figure \ref{fig:density2} we plot the approximate transition density $p^{(N)}$ for a CEV-like model with Gaussian jumps
\begin{align}
\nu_i(dz)
        &=              \frac{1}{\sqrt{ 2 \pi s_i^2}} \exp \( \frac{-(z-m_i)^2}{2 s_i^2} \) dz . \label{eq:gaussian}
\end{align}
For the smallest initial value in the plot, $y = -0.6$ we see that $p^{(8)}$ and $p^{(9)}$ are virtually identical.  As the initial value $y$ moves in the positive direction, fewer terms are required for convergence.  For $y=0.0$, we see very little difference between  $p^{(4)}$ and $p^{(5)}$.  And for $y=0.6$, we see that $p^{(2)}$ and $p^{(3)}$ are nearly identical.  This is not surprising, since the size of the perturbing term $\eps e_\beta \Ac_1$ decreases as $y$ tends to infinity.

\subsubsection{Comparison to Monte Carlo simulation}
\label{sec:monte.carlo}
In order to test the accuracy of pricing formula \eqref{eq:price.CEV} we compute the price of a series of call options with $N=10$.  We once again assume Gaussian jumps, as in \eqref{eq:gaussian}.  For each call option, we also compute its price using Monte Carlo simulation.  For the Monte Carlo simulations we use a standard Euler scheme with a time step of $10^{-3}$ years and run $10^7$ sample paths.  As implied volatility -- rather than price -- is the more relevant quantity for call options, we convert prices to implied volatilities by inverting the Black-Scholes formula numerically (we examine our implied volatility expansion in the next section).  In Figure \ref{fig:impvol} we plot the resulting implied volatilities as a function of the $\log$-moneyness to maturity ratio, $\text{LMMR}:=(k-y)/t$.  For the strikes and maturities tested, we see very close agreement between the implied volatilities resulting from pricing approximation \eqref{eq:price.CEV} and the implied volatilities resulting from the Monte Carlo simulation.

\subsubsection{Implied Volatility Expansion}
In section we examine the implied volatility expansion of Section \ref{sec:impvol}.  We continue to work in the CEV-like setting with $\eta = e_\beta$.  But, we now set $\nu_0 \equiv 0$ and $c_0=0$, which is an assumption of Section \ref{sec:impvol}.  We still assume $\nu_1$ is Gaussian, as in equation \eqref{eq:gaussian}.  We define the $\Oc(\eps^n)$ approximation of the implied volatility
\begin{align}
\sig^{(n)}
        &:=             \sum_{k=0}^n \eps^k \sig_k , \label{eq:sigma.approximate}
\end{align}
where $\sig_0=a_0$ and the $\{\sig_k\}_{n=1}^\infty$ are given by \eqref{eq:sig.k}.  The values of $u_n$, which are needed for the implied volatility expansion, are computed using \eqref{eq:u.n.2}.  In figure \ref{fig:impvol2} we plot $\sig^{(n)}$ for $n=0,1,\cdots,5$.  In order to see how well the truncated implied volatility expansion approximates the exact implied volatility $\sig^\eps$ we also plot a proxy of $\sig^\eps$.  Our proxy for $\sig^\eps$ is obtained by approximating $u^\eps$ with $u^{(12)}$, and then  by inverting the Black-Scholes formula numerically to obtain $\sig^\eps$.  The price $u^{(12)}$ is computed using \eqref{eq:price.CEV}.  Given the numerical results of Section \ref{sec:monte.carlo}, approximating $u^\eps$ with $u^{(12)}$ should not introduce much error.
\par
In Figure~\ref{fig:impvol2} we see very fast convergence of $\sig^{(n)}$ to $\sig^\eps$ for $\text{LMMR} \in [-0.5,3.0]$.  In this region $\sig^{(3)}$ is nearly indistinguishable from $\sig^\eps$.  Outside of this region, however, the implied volatility expansion does not converge.  This is due to the fact that, for $\text{LMMR} \notin [-0.5,3.0]$ we have $|u^\eps-u^{\BS}(a_0)| > R$, where $R=R(t,y,k,a_0)$ is the radius of convergence of the infinite series \eqref{eq:inverse} with $\sig_0 = a_0$.

\subsubsection{Calibration to S{\&}P500 options}
\label{sec:sp500}
In order to demonstrate the applicability of the CEV-like models from Section~\ref{sec:example} we perform a sample calibration to S{\&}P500 options.
For the calibration, we assume that jumps are Gaussian, i.e. that 
$$
\nu_i(dz) = \frac{\Gamma_i}{\sqrt{2\pi s^2}}\exp\( \frac{(y-m)^2}{2 s^2} \) dz,
$$
for $i=1,2$.
We have assumed here a common mean $m$ and variance $s^2$, but have allowed for separate jump intensities $\Gamma_0, \Gamma_1>0$.
Thus the jump distribution remains constant, but the intensity $\Gamma_0 + \eps e^{\beta y}\Gamma_1$ 
varies with~$y$.  
One could allow for additional flexibility by considering separate means and variances.  
\par
Let $\Phi$ be the set of model parameters and let $\Theta$ be the feasible set for these parameters.  We denote by $\text{IV}(t,k;\Phi)$ the implied volatility of an option with time to maturity $t$ and $\log$-strike $k$, as computed using~$\Phi$, and we denote by $\text{IV}^{obs}(t,k)$, the observed implied volatility of an option with time-to-maturity $t$ and $\log$-strike $k$.  We formulate the calibration problem as a least-squares fit to the observed implied volatility.  That is, we seek $\Phi^{*}$ such that  
\begin{align}
\inf_{\Phi \in \Theta} \sum_{{(t,k)\in (\mathcal{T},\mathcal{K})}} \( \text{IV}^{obs}(t,k) - \text{IV}(t,k;\Phi) \)^2 
        &=       \sum_{{(t,k)\in (\mathcal{T},\mathcal{K})}} \( \text{IV}^{obs}(t,k) - \text{IV}(t,k;\Phi^{*}) \)^2,
\end{align}
where $(\mathcal{T},\mathcal{K})$ represents the set of all (maturity, strike) observed implied volatility data.
Observe that we fit all maturities in the data set simultaneously; we do \emph{not} fit maturity-by-maturity.  Note, because $\nu_0 \neq 0$ and $c_0 \neq 0$, we are not in the setting of Section \ref{sec:impvol}.  Thus we must compute implied volatilities by first computing option prices using \eqref{eq:price.CEV}, and then by inverting the Black-Scholes formula numerically.  The results of the calibration procedure are plotted in Figure \ref{fig:SP500}.  The figure clearly shows that the CEV-like model considered in this section provides a tight fit to implied volatility across maturities.
\par 
Using the parameters obtained in the calibration procedure, we run a series of numerical tests in order to investigate the computational cost of computing $\text{IV}^{(N)}$ (the implied volatility induced by $u^{(N)}$) for different values of $N$.  As a point of comparison, we note that $u^{(0)}$ corresponds to the price of an option as computed in an exponential L\'evy setting (i.e., an exponential L\'evy model with no local-dependence).  As demonstrated in Table \ref{tab:IV},  for $N=3$ we obtain we obtain implied volatilities that are accurate to two decimal places.  These implied volatilities require roughly $2.22$ times as long to compute as the corresponding implied volatilities in an exponential L\'evy setting.

\subsubsection{Implied volatility asymptotics for CEV-like models with no jumps}
\label{sec:impvol.3}
In our last numerical experiment, we implement the implied volatility expansion outlined 
in Section~\ref{sec:impvol.2}.
Under Assumption \eqref{ass:phi.chi} we compute approximate implied volatilities $\sig^{(2,M)}$ 
using~\eqref{eq:sig.2.M} and~\eqref{eq:sig.2.M.b}.
For comparison, we also plot the exact implied volatility $\sig^\eps$.
To compute $\sig^\eps$, we first compute $u^\eps$ using~\eqref{eq:price.CEV} and then 
we invert the Black-Scholes formula numerically.  
The results are plotted in Figure~\ref{fig:IV}.  
With a time-to-maturity of $t=1/2$, 
we observe a nearly exact match between $\sig^{(2,M)}$ and $\sig^\eps$ for $\log$-moneyness $k-y>-0.5$.

%
%

\section{Conclusion}
\label{sec:conclusion}
In this paper we introduce a class of L\'evy-type models in which the diffusion coefficient, the L\'evy measure and the default intensity all depend locally on the value of the underlying.  Within this framework, we obtain a formula (written as an infinite series) for the price of a European option.  Furthermore, we provide conditions under which this infinite series is guaranteed to converge.  Additionally, we obtain an explicit expression for the implied volatility smile induced by a certain sub-class of L\'evy-type models.  This series is exact within its radius of convergence.  As an example of our framework, we introduce a CEV-like L\'evy-type model, which corrects some short-comings of the CEV model; namely (i) our choice of local volatility function does not drop to zero as the value of the underlying increases and (ii) our model permits the underlying asset to experience jumps.  In this CEV-like setting, we show that option prices can be computed with the same level of efficiency as other models in which option prices are computed as Fourier-type integrals and we show that approximate implied volatilities can be computed explicitly without integration.  We also test the accuracy of the pricing and implied volatility formulas in the CEV-like setting through numerical examples.  And we show that one specific CEV-like model with normal jumps provides a tight fit to the observed S{\&}P500 implied volatility surface.

\subsection*{Thanks}
The authors would like to thank Bjorn Birnir, Stephan Sturm and two anonymous referees for their helpful comments.

%
%

\bibliographystyle{chicago}
\bibliography{BibTeX-Master}    

%
%

\clearpage
\appendix

%
%

\section{Proof of Proposition \ref{thm:un}}
\label{sec:un}
We begin the proof by Fourier transforming the left-hand side of \eqref{eq:u0.PIDE} and \eqref{eq:un.PIDE}.  We have
\begin{align}
\< \psi_\lam, (-\d_t + \Ac_0) u_n \>
        &=      -\d_t \< \psi_\lam, u_n \> + \< \psi_\lam, \Ac_0 u_n \>
        =               -\d_t \< \psi_\lam, u_n \> + \< \Ad_0 \psi_\lam, u_n \>
        =               ( -\d_t + \phi_\lam ) \< \psi_\lam , u_n \> ,   
\end{align}
where we use $\Ad_0 \overline{\psi_\lam} = \phi_\lam \overline{\psi_\lam}$.  
Fourier transforming the right-hand side of \eqref{eq:un.PIDE} and the initial conditions yields the following ODEs in the variable $t$ for $\<\psi_\lam,u_0\>$ and for the sequence $(\<\psi_\lam,u_n\>)_{n \geq 1}$:
\begin{align}
\Oc(1):&&
( -\d_t + \phi_\lam ) \< \psi_\lam , u_0 \>
        &= 0, &
\< \psi_\lam , u_0(0,\cdot) \>
        &=      \< \psi_\lam , h \> , \\
\Oc(\eps^n):&&
( -\d_t + \phi_\lam ) \< \psi_\lam , u_n \>
        &=      - \< \psi_\lam , \eta \Ac_1 u_{n-1} \> , &
\< \psi_\lam , u_n(0,\cdot) \>
        &=      0,\qquad n\geq 1 .
\end{align}
The following solutions can easily be checked (e.g., by substitution)
\begin{align}
\Oc(1):&&
\< \psi_\lam , u_0(t,\cdot) \>
        &=      e^{t \phi_\lam} \< \psi_\lam , h \> , \\
\Oc(\eps^n):&&
\< \psi_\lam , u_n(t,\cdot) \>
        &=      \int_0^t ds e^{(t-s) \phi_\lam}  \< \psi_\lam , \eta \Ac_1 u_{n-1}(s,\cdot) \>,\qquad n\geq 1  .
\end{align}
Next, using~\eqref{eq:identity}, we obtain
\begin{align}
\Oc(1):&&
u_0(t,y)
        &=      \int_\Rb d\lam e^{t \phi_\lam} \< \psi_\lam , h \> \psi_\lam(y) , \\
\Oc(\eps^n):&&
u_n(t,y)
        &=      \int_\Rb d\lam \int_0^t ds e^{(t-s) \phi_\lam} \< \psi_\lam , \eta \Ac_1 u_{n-1}(s,\cdot) \> \psi_\lam(y),\qquad n\geq 1 .
\end{align}
Note that the sequence $(u_n)_{n\geq 1}$ can be computed recursively.
For example,
\begin{align}
u_1(t,y)
        &=      \int_\Rb d\lam \int_0^t ds e^{(t-s) \phi_\lam} \< \psi_\lam , \eta \Ac_1 u_{0}(s,\cdot) \> \psi_\lam(y) \\
        &=      \int_\Rb \int_\Rb d\lam d\mu \int_0^t ds e^{(t-s) \phi_\lam}
                        \< \psi_\lam , \eta \Ac_1 e^{s \phi_\mu} \< \psi_\mu , h \> \psi_\mu \> \psi_\lam(y) \\
        &=      \int_\Rb \int_\Rb d\lam d\mu \int_0^t ds e^{t\phi_\lam + s(\phi_\mu-\phi_\lam)} \chi_\mu
                        \< \psi_\lam , \eta \psi_\mu \> \< \psi_\mu , h \>\psi_\lam(y) \\
        &=      \int_\Rb \int_\Rb d\lam d\mu \( \frac{e^{t\phi_\mu} - e^{t\phi_\lam}}{\phi_\mu - \phi_\lam} \) \chi_\mu
                        \< \psi_\lam , \eta \psi_\mu \> \< \psi_\mu , h \>\psi_\lam(y)
\end{align}
Generalizing the above recursion relation to arbitrary $n$, one finds expression \eqref{eq:u.n} for $u_n$.

%
%

\section{Proof of Theorem \ref{thm.u.eps}}
\label{sec:proof2}
In this section, we will show that $u^\eps$, given by~\eqref{eq:u.eps.expand} and~\eqref{eq:u.n}, is a classical solution of the Cauchy problem~\eqref{eq:u.eps.PIDE} under the conditions of Theorem \ref{thm.u.eps}.  Throughout this section, we define a Hilbert space $\Hc=L^2(\Rb)$ with norm $\norm{\cdot} = \<\cdot,\cdot\>$ given by~\eqref{eq:L2.norm}.  
Our strategy is to show that the closure of~$\Ac^\eps=\Ac_0 + \eps \eta \Ac_1$ 
(with a domain restricted to $\Hc$) generates a~$C_0$-contraction semigroup~$\Pc^\eps=\{\Pc_t^\eps, t \geq 0 \}$ in~$\Hc$.  
The semigroup~$\Pc^\eps$ has the property that if $h \in \dom{\Ac^\eps}$ then $\Pc_t^\eps h \in \dom{\Ac^\eps}$ (\cite{engel2006short}, Proposition II.6.2) and
$\left(-\d_t + \Ac^\eps \right)\Pc_t^\eps h =0$ with initial condition~$\Pc_0^\eps h = h$.
Thus, if we can show that $\Ac^\eps$ generates a semigroup~$\Pc^\eps$, we can identify $u^\eps(t,y) \equiv \Pc^\eps h(t,y)$ as the unique classical solution to~\eqref{eq:u.eps.PIDE}.  
Moreover, if it exists, the semigroup~$\Pc_t^\eps$ is given by
$\Pc_t^\eps = \exp ( t \Ac^\eps ) = \exp ( t [\Ac_0 + \eps \eta \Ac_1])$,
where the exponential is defined by
\begin{align}
\exp( t \Ac^\eps )
        &:=             \lim_{n \to \infty} \( 1 - \frac{t}{n} \Ac^\eps \)^{-n} ,
\end{align}
and the solution $u^\eps(t,y) = \Pc^\eps h(t,y)$ \emph{inherits} the analyticity of the exponential in the perturbing parameter~$\eps$.  
Thus, if $\Ac^\eps$ generates a semigroup $\Pc^\eps$, then~$u^\eps$ is an analytic function of~$\eps$, 
and has the representation~\eqref{eq:u.eps.expand}.
\par
We start by defining the domains of the operators $\Ac_0$, $\Ac_1$ and $\eta \Ac_1$:
$\dom{\Ac_i} := \{ u \in \Hc : \norm{\Ac_i u } < \infty \}$ for $i=0,1$
and $\dom{\eta \Ac_1} := \{ u \in \Hc : \norm{\eta \Ac_1 u } < \infty \}$.
Note that 
\begin{align}
\left\| \Ac_0 u \right\|^2
        &=              \int_\Rb d\lam |\<\psi_\lam,u\>|^2 |\phi_\lam|^2 , &
\left\| \Ac_1 u \right\|^2
        &=              \int_\Rb d\lam |\<\psi_\lam,u\>|^2 |\chi_\lam|^2 , &
\left\| \eta \right\|^2
        &=              \int_\Rb d\lam |\<\psi_\lam,\eta\>|^2 .
\end{align}
Thus by~\eqref{eq:implies}, we have $\dom{\Ac_0} \subseteq \dom{\Ac_1}$.  
Since $\eta \in \Sc$, then $\norm{\eta \Ac_1 u}^2 \leq \norm{\eta}^2 \cdot \norm{ \Ac_1 u}^2$
is finite for any $u \in \dom{\Ac_1}$.  
Therefore the inclusions~$\Sc \subseteq \dom{\Ac_0} \subseteq \dom{\Ac_1} \subseteq \dom{\eta \Ac_1}$ hold.
Therefore since $\Sc$ is a dense subset of $\Hc$ (see~\citet*{jacob2001pseudo}, Corollary 2.6.1), 
the operators 
$\Ac_0$, $\Ac_1$ and $\eta \Ac_1$ are densely defined in~$\Hc$.  
To show that $\Ac^\eps$ generates a semigroup $\Pc^\eps$ we recall the following theorem from \citet*{chernoff}:
\begin{theorem}
\label{thm:semigroup}
Let $\Ac$ be the generator of a $C_0$-contraction semigroup $\Pc_t^0\equiv \exp(t \Ac)$ on a Banach space,
and~$\eps \Bc$ a dissipative operator with a densely defined adjoint.
If there exist two real constants $A \geq 0$ and $B \leq 1$ such that the inequality
$\norm{\eps \Bc u} \leq A \norm{u} + B \norm{ \Ac u }$ 
holds for all $u\in {\rm dom}(\Ac)$ (i.e., the operator $\eps \Bc$ is bounded relative to~$\Ac$ with relative bound~$B$), 
then the closure of $\Ac^\eps:=\Ac+\eps\Bc$ 
generates a $C_0$-contraction semigroup~$\Pc_t^\eps=\exp(t \Ac^\eps)$.
\end{theorem}
\noindent
We now check the conditions of Theorem \ref{thm:semigroup} with $\Ac_0\equiv \Ac$ 
and $\eps \eta \Ac_1\equiv \eps \Bc$.
First, by Corollary II.3.17 in~\citet*{engel2006short}, $\Ac_0$,
as the generator of a L\'evy process that is an exponentially special semimartingale, 
generates a $C_0$-contraction semigroup.  
Next, by Theorem 2.12 in~\citet*{hoh1998pseudo}, $\eps \eta \Ac_1$ satisfies the positive maximum principle, 
and hence is dissipative (\citet*{ethier1986markov}, Lemma 4.2.1).
Since~$\eta \in \Sc$ and since Hilbert spaces are reflexive, then $\eps \eta \Ac_1$ has a densely defined adjoint (see the discussion after the main theorem in~\citet*{chernoff}).
Theorem \ref{thm.u.eps} will then follow if we can prove that~$\eps \eta \Ac_1$ 
is bounded relative to $\Ac_0$ with relative bound $B$.
\begin{proposition}
\label{thm:relative.bound}
Suppose that there exist two constants $A \geq 0$ and $B \leq 1$ such that $\eps$ satisfies
\begin{align}
\eps^2
        &\leq    \inf_{\lam \in \Rb} \frac{A^2 + B^2 |\phi_\lam|^2}{\left\|\eta \right\|^2 \cdot | \chi_\lam |^2}.
\end{align}
Then $\eps \eta \Ac_1$ is bounded relative to $\Ac_0$ with relative bound $B$.
\end{proposition}
\begin{proof}
For any $u \in \dom{\Ac_0}$, the inequality in the proposition holds if and only if 
$\eps^2 \leq  \frac{A^2 + B^2 |\phi_\lam|^2}{\left\|\eta \right\|^2 \cdot | \chi_\lam |^2}$
holds for all $\lambda\in\Rb$.
This in turn is equivalent to
$0 \leq A^2 + B^2 |\phi_\lam|^2 -  \eps^2 \left\|\eta \right\|^2 \cdot | \chi_\lam |^2$ 
for any $\lambda\in\Rb$, which implies that
$$
0 \leq \int_\Rb d\lam  |\<\psi_\lam,u\>|^2 \( A^2 + B^2 |\phi_\lam|^2 -  \eps^2 \left\|\eta \right\|^2 \cdot | \chi_\lam |^2 \) 
=              A^2 \left\| u \right\|^2 + B^2 \left\| \Ac_0 u \right\|^2 - \eps^2 \left\|\eta \right\|^2 \cdot \left\| \Ac_1 u \right\|^2.
$$
This then implies 
$\left\|\eps \eta \Ac_1 u \right\|^2 \leq A^2 \left\| u \right\|^2 + B^2 \left\| \Ac_0 u \right\|^2$.
Since  $\left\| \eta \Ac_1 u \right\| \leq \left\|\eta \right\| \cdot \left\| \Ac_1 u \right\|$, 
we then deduce the final inequality
$\left\|\eps \eta \Ac_1 u \right\|  \leq A \left\| u \right\| + B \left\| \Ac_0 u \right\|$,
and the proposition follows.
\end{proof}
We have now shown that $\Ac^\eps$ generates a semigroup $\Pc^\eps$.  Therefore, we identify $u^\eps(t,y) = \Pc_t^\eps h(y)$ and we note that $u^\eps(t,y)$ is analytic in the perturbing parameter $\eps$.

%
%

\section{Mathematica code for computing $\sig^{(2,M)}$}
\label{sec:mathematica}
The following code will produce $\sig^{(2,M)}$ from equation \eqref{eq:sig.2.M.b} with $M=10$.
\begin{align}
M
        &=      10;\\
\phi [\lambda \_]
        &=      \frac{1}{2} \text{a0}^2 \left(-\lambda ^2-i \lambda \right);
\\
\chi [\lambda \_]
        &=      \frac{1}{2} \text{a1}^2 \left(-\lambda ^2-i \lambda \right);
\\
\text{b1}[\text{t$\_$},\text{a0$\_$},\text{a1$\_$}]
        &=\text{CoefficientList}\left[\sum _{n=1}^M \frac{t{}^{\wedge}n}{n!}(\phi [-i d - i \beta ]-\phi [-i d]){}^{\wedge}(n-1)\frac{1}{2} \text{a1}^2,d\right];
\\
\text{b2}[\text{t$\_$},\text{a0$\_$},\text{a1$\_$}]
        &=      \text{CoefficientList}\Big[ 
                        \sum _{n=2}^M \frac{t{}^{\wedge}n}{n!}\sum _{k=1}^{n-1} \text{Binomial}[n-1,k](-\phi [-i d]){}^{\wedge}(n-1-k) \\ &\qquad
                        \sum _{m=0}^{k-1} (\phi [-i d - i \beta ]){}^{\wedge}m (\phi [-i d - 2i \beta ]){}^{\wedge}(k-1-m)\chi [-i d - i \beta ]\frac{1}{2} \text{a1}^2, d \Big];
\\
\text{dp}[\text{t$\_$},\text{y$\_$},\text{a0$\_$},\text{k$\_$}]
        &=      (y-k+(\text{a0}{}^{\wedge}2/2)t)/(\text{a0} \sqrt{t});
\\
\text{$\sigma $1}[\text{t$\_$},\text{y$\_$},\text{a0$\_$},\text{a1$\_$},\beta \_,\text{k$\_$}]
        &=      \text{Exp}[\beta  y]\sum _{n=0}^{M-1} \text{b1}[t,\text{a0},\text{a1}][[n+1]] \\ &\qquad \text{FullSimplify}\left[\frac{D[\text{Exp}[y-\text{dp}[t,y,\text{a0},k]{}^{\wedge}2/2],\{y,n\}]}{t \, \text{a0} \, \text{Exp}[y-\text{dp}[t,y,\text{a0},k]{}^{\wedge}2/2]}\right];
\\
\text{$\sigma $1}[\text{t$\_$},\text{y$\_$},\text{a0$\_$},\text{a1$\_$},\beta \_,\text{k$\_$}]
        &=      \text{Exp}[2 \beta  y]\sum _{n=0}^{M-1} \text{b2}[t,\text{a0},\text{a1}][[n+1]] \\ &\qquad \text{FullSimplify}\left[\frac{D[\text{Exp}[y-\text{dp}[t,y,\text{a0},k]{}^{\wedge}2/2],\{y,n\}]}{t \, \text{a0} \, \text{Exp}[y-\text{dp}[t,y,\text{a0},k]{}^{\wedge}2/2]}\right] \\ &\qquad
-\frac{1}{2}\Big( \text{$\sigma $1}[t,y,\text{a0},\text{a1},\beta ,k] \Big){}^{\wedge}2\left(\frac{(k-y){}^{\wedge}2}{t \text{a0}{}^{\wedge}3}-\frac{t \, \text{a0}}{4}\right);
\\
\sigma{2M} [\text{t$\_$},\text{y$\_$},\text{a0$\_$},\text{a1$\_$},\beta \_,\eps \_,\text{k$\_$}]
        &=      \text{a0}+\eps\,  \text{$\sigma $1}[t,y,\text{a0},\text{a1},\beta ,k]+\eps{}^{\wedge}2 \, \text{$\sigma $2}[t,y,\text{a0},\text{a1},\beta ,k];
\end{align}

%
%

\clearpage
\begin{figure}[t]
\centering
\begin{tabular}{ | c | c | c |}
\hline
$n=1$ & $n=2$ & $n=3$ \\
\includegraphics[width=.3\textwidth,height=.17\textheight]{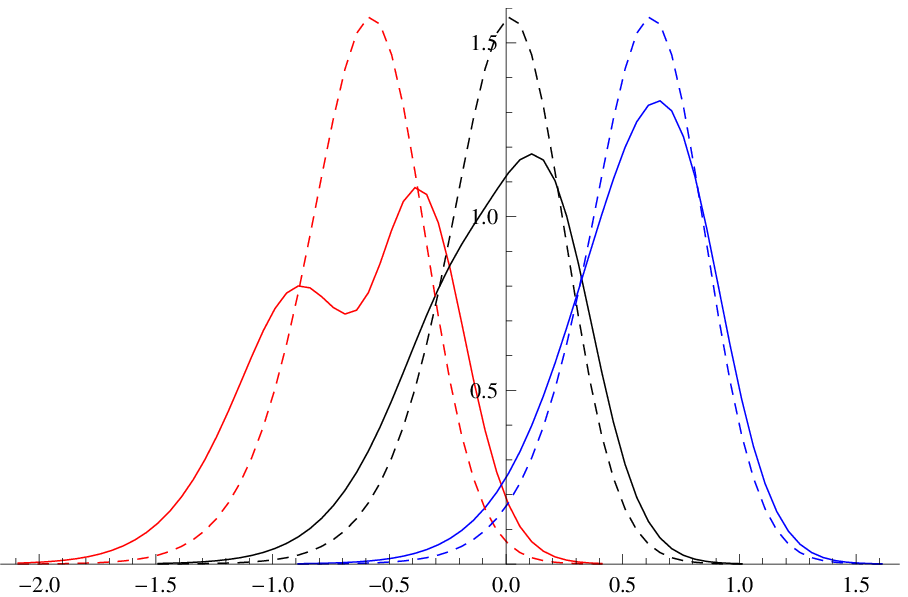} &
\includegraphics[width=.3\textwidth,height=.17\textheight]{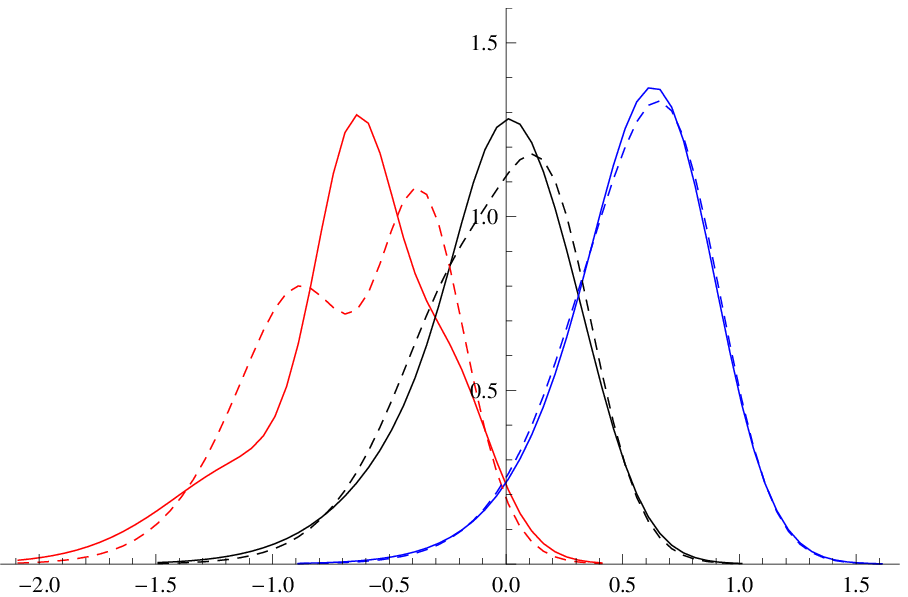} &
\includegraphics[width=.3\textwidth,height=.17\textheight]{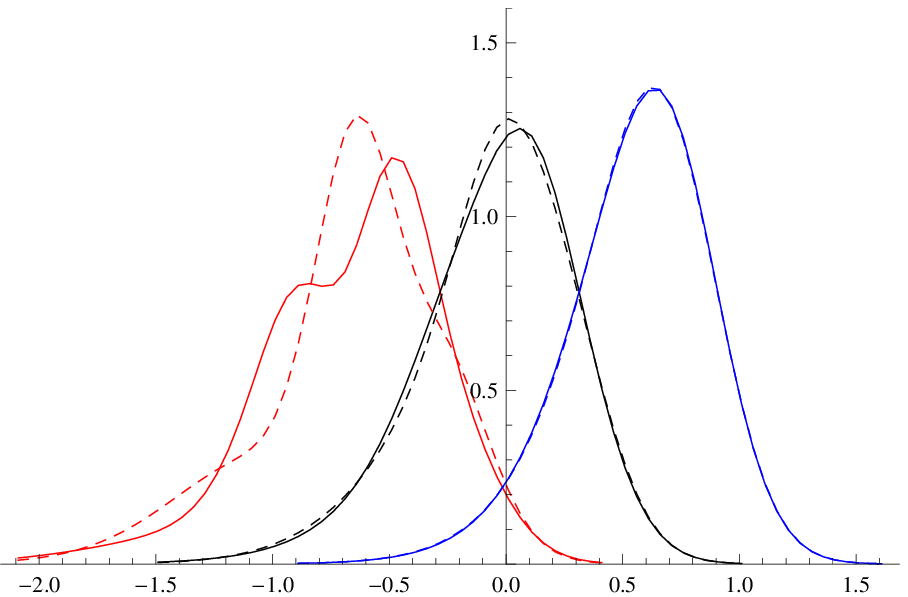} \\ \hline 
$n=4$ & $n=5$ & $n=6$ \\
\includegraphics[width=.3\textwidth,height=.17\textheight]{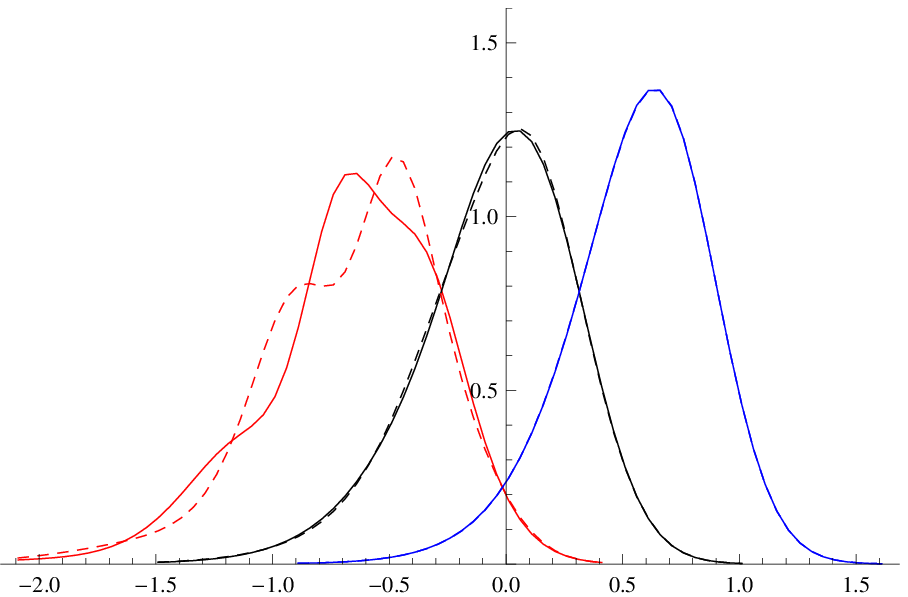} &
\includegraphics[width=.3\textwidth,height=.17\textheight]{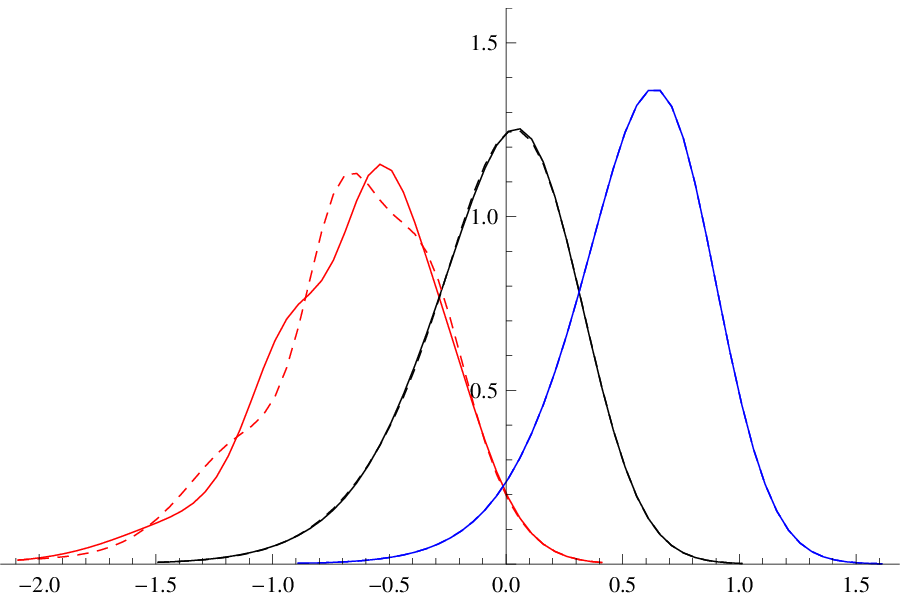} &
\includegraphics[width=.3\textwidth,height=.17\textheight]{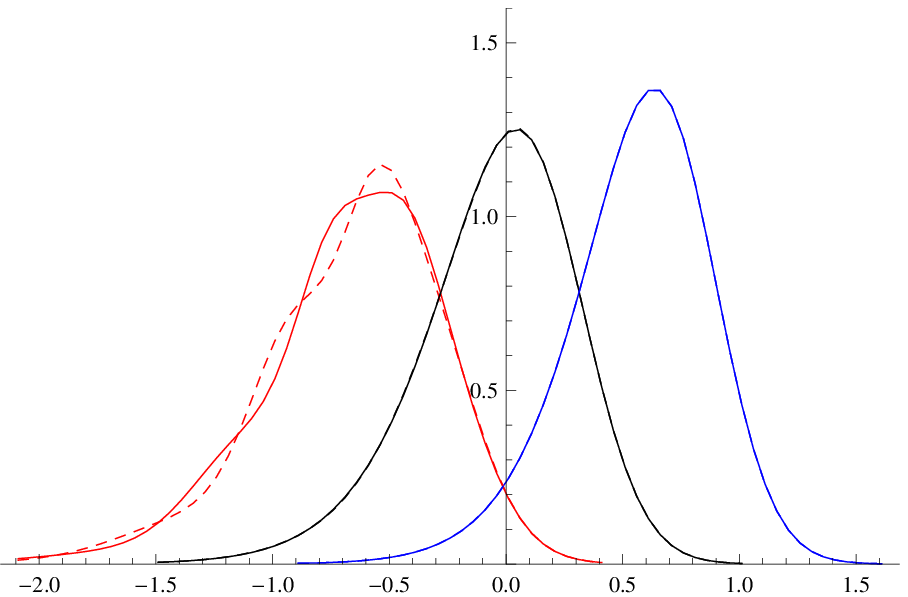} \\ \hline 
$n=7$ & $n=7$ & $n=9$ \\
\includegraphics[width=.3\textwidth,height=.17\textheight]{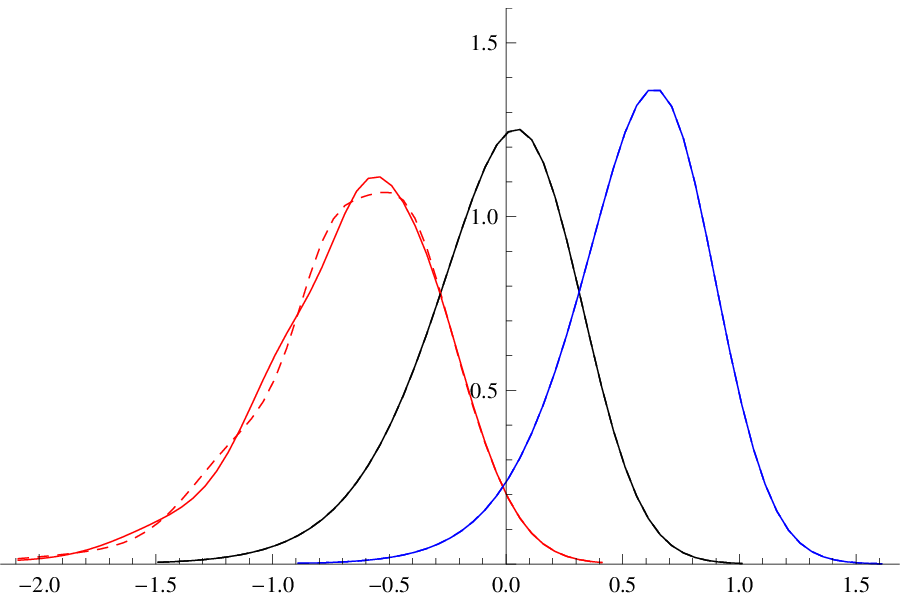} &
\includegraphics[width=.3\textwidth,height=.17\textheight]{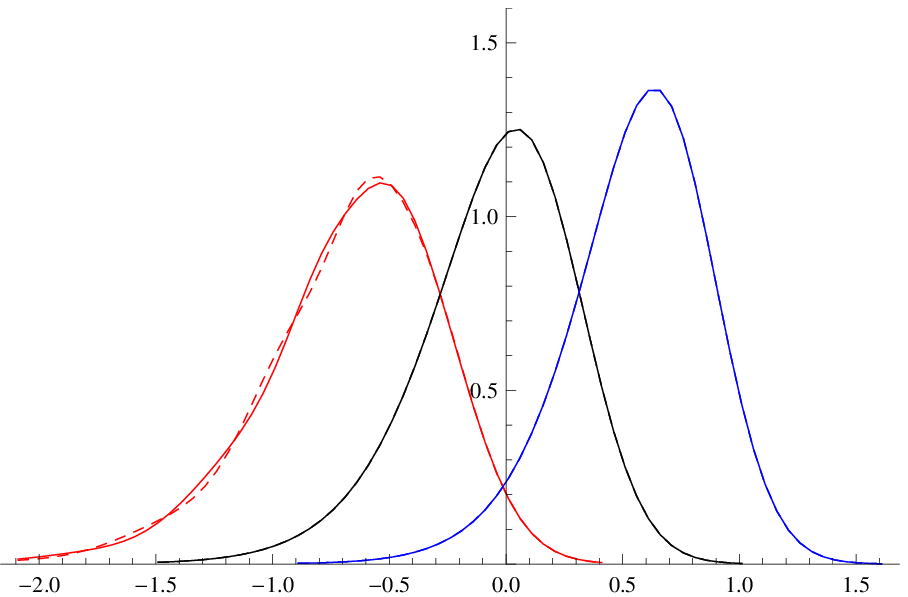} &
\includegraphics[width=.3\textwidth,height=.17\textheight]{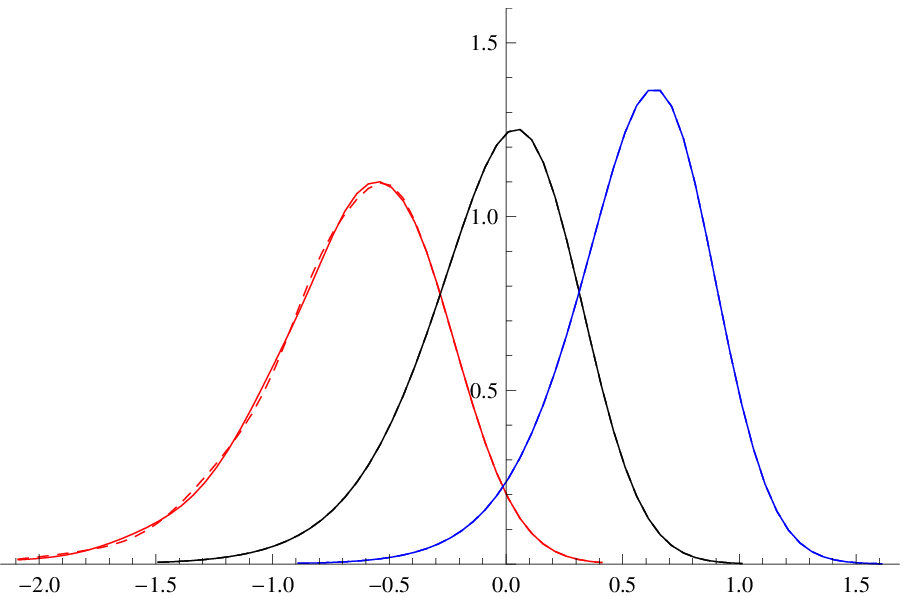} \\ \hline
\end{tabular}
\begin{align}
\nu_i(dz)
        &=              \frac{1}{\sqrt{ 2 \pi s_i^2}} \exp \( \frac{-(z-m_i)^2}{2 s_i^2} \) dz
\end{align}
\caption{For different values of $n$ we plot as a function of $z$ the approximate transition 
densities~$p^{(n)}(t,y,z)$ for $y=-0.6$ (solid red), $y=0.0$ (solid black) and $y=0.6$ (solid blue).  
In order to see the convergence, on each plot, we also graph $p^{(n-1)}(t,y,z)$ (dashed lines).  
Note that as $y$ moves closer to $-\infty$ (i) the transition densities become wider, (ii) convergence of the densities requires more terms and (iii) the densities have fatter tails on the left than on the right.  
All three phenomena are due to the fact that the local volatility and the jump-intensity rise as $y$ 
decreases to $-\infty$.  
The following parameters are used in these plots: $a_0=0.20$, $a_1=0.10$, $c_0=0.0$, $c_1=0.0$, $s_0=0.15$, $m_0=-0.10$, $s_1=0.15$, $m_1=-0.10$, $\eps=1$, $\beta=-0.95$, $t=1.0$.
}
\label{fig:density2}
\end{figure}


\clearpage
\begin{figure}[t]
\centering
\begin{tabular}{ | c | c |}
\hline
$t=0.125$ & $t=0.25$ \\
\includegraphics[width=.475\textwidth,height=.3\textheight]{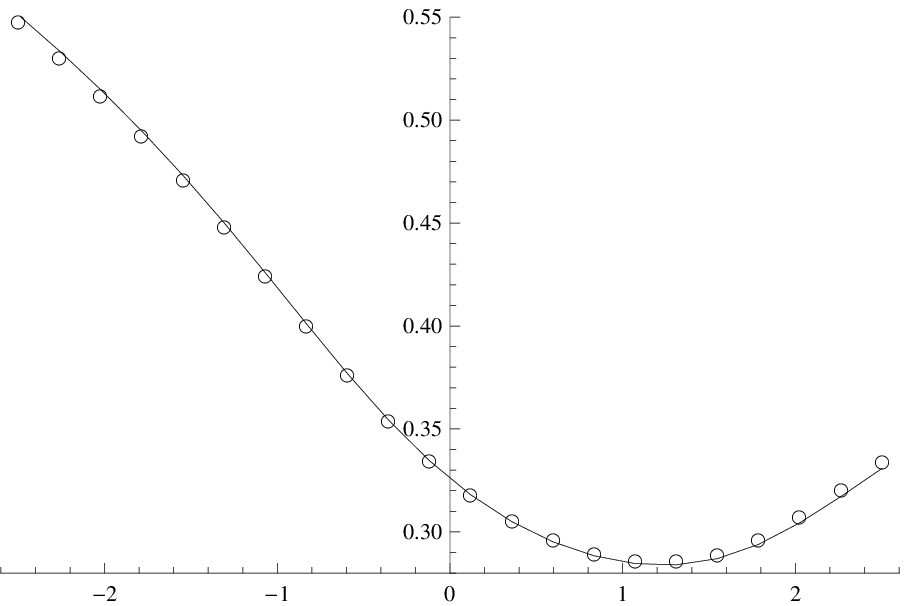} &
\includegraphics[width=.475\textwidth,height=.3\textheight]{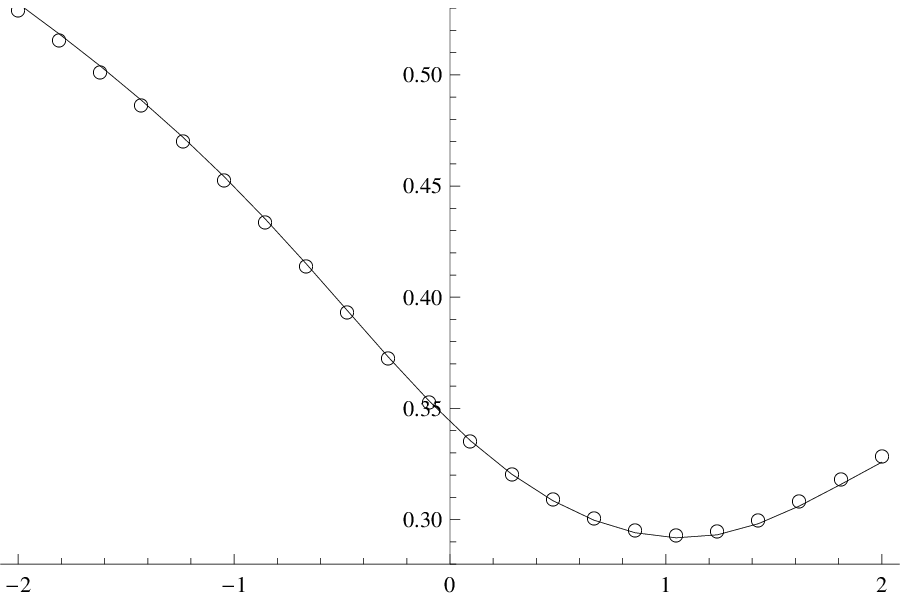} \\ \hline 
$t=0.5$ & $t=1.0$ \\
\includegraphics[width=.475\textwidth,height=.3\textheight]{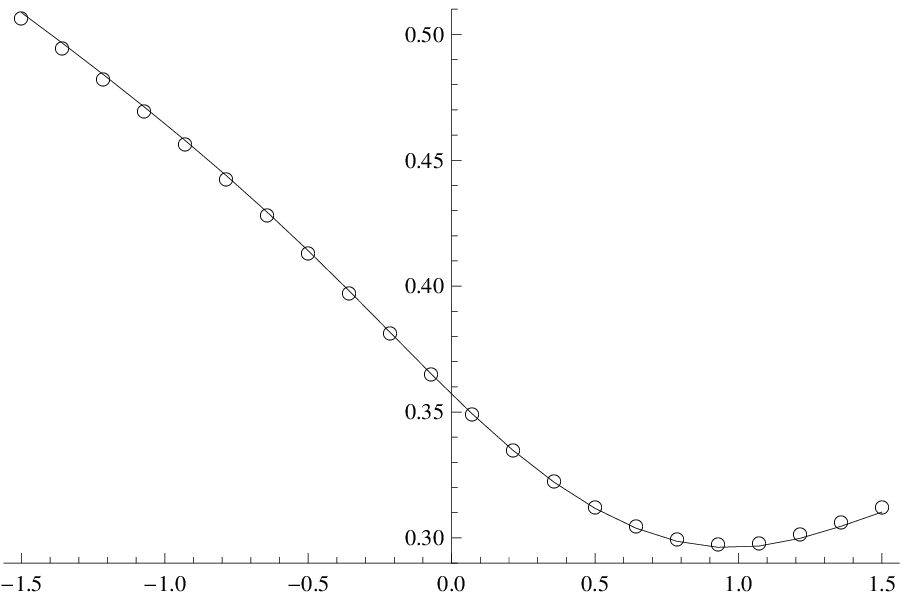} &
\includegraphics[width=.475\textwidth,height=.3\textheight]{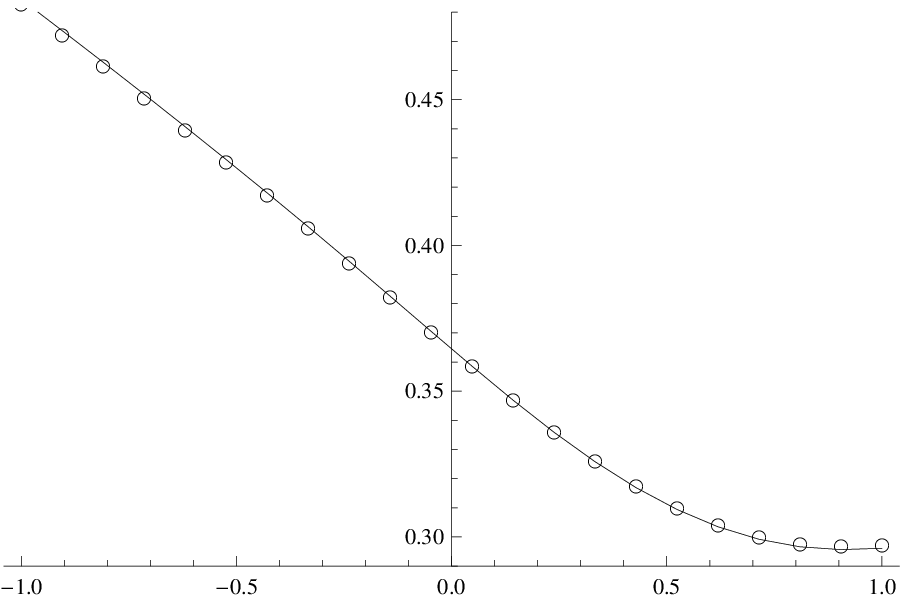} \\ \hline 
\end{tabular}
\begin{align}
\nu_i(dz)
        &=              \frac{1}{\sqrt{ 2 \pi s_i^2}} \exp \( \frac{-(z-m_i)^2}{2 s_i^2} \) dz
\end{align}
\caption{In the above plots, we compute option prices using formula~(\ref{eq:price.CEV})
with $N=10$ and also by Monte Carlo simulation.  
We then convert these prices to implied volatilities by inverting Black-Scholes numerically.  
We do \emph{not} use the implied volatility expansion described in Section \ref{sec:impvol}.  
The solid line corresponds to implied volatilities computed using pricing formula (\ref{eq:price.CEV}).  
The circles correspond to implied volatilities resulting from the Monte Carlo simulation.  
Units on the horizontal axis are logmoneyness to maturity ratios ($\text{LMMR}:=(k-y)/t$).  
Note the steep skew, which is due to the fact that the local volatility and the jump intensity increase as the value of the underlying drops.  
The following parameters are used in this plot: $a_0=0.20$, $a_1=0.10$, $c_0=0.00$, $c_1=0.00$, $s_0=0.20$, $s_1 = 0.10$, $m_0=-0.20$, $m_1 = -0.10$, $\eps=1.0$, $\beta=-1.25$, $y=-0.10$.}
\label{fig:impvol}
\end{figure}


\clearpage
\begin{figure}
\centering
\begin{tabular}{ | c | c |}
\hline
$n=0$ & $n=1$ \\
\includegraphics[width=.465\textwidth,height=.24\textheight]{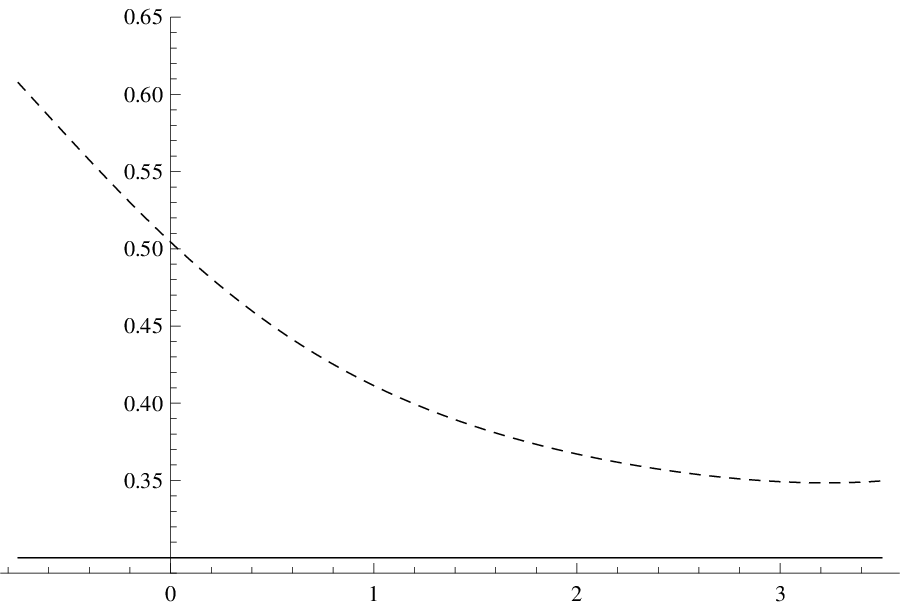} &
\includegraphics[width=.465\textwidth,height=.24\textheight]{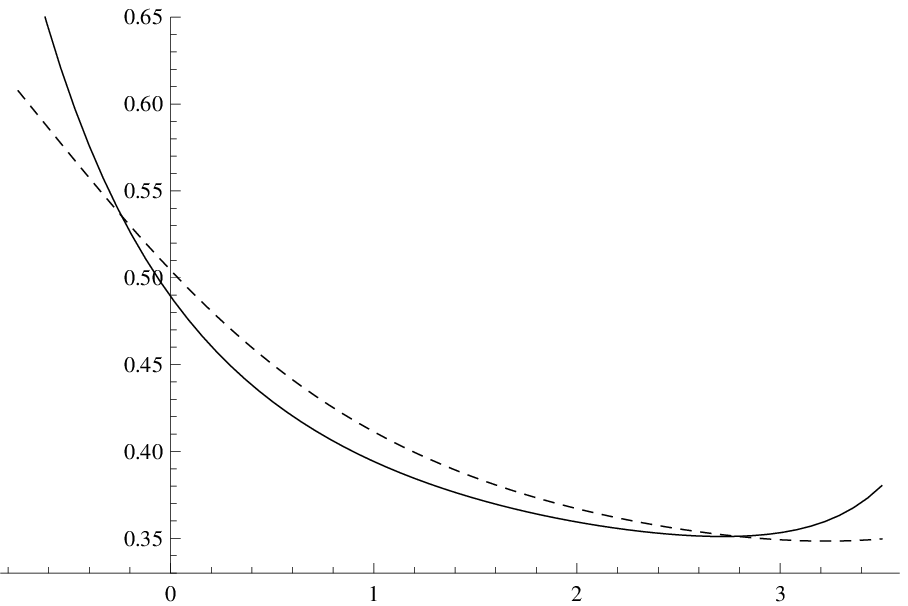} \\ \hline 
$n=2$ & $n=3$ \\
\includegraphics[width=.465\textwidth,height=.24\textheight]{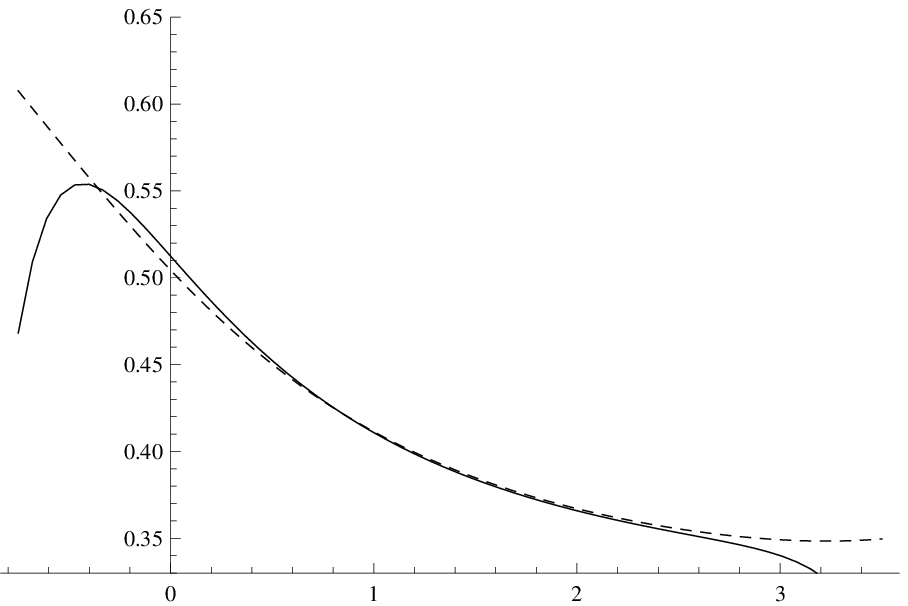} &
\includegraphics[width=.465\textwidth,height=.24\textheight]{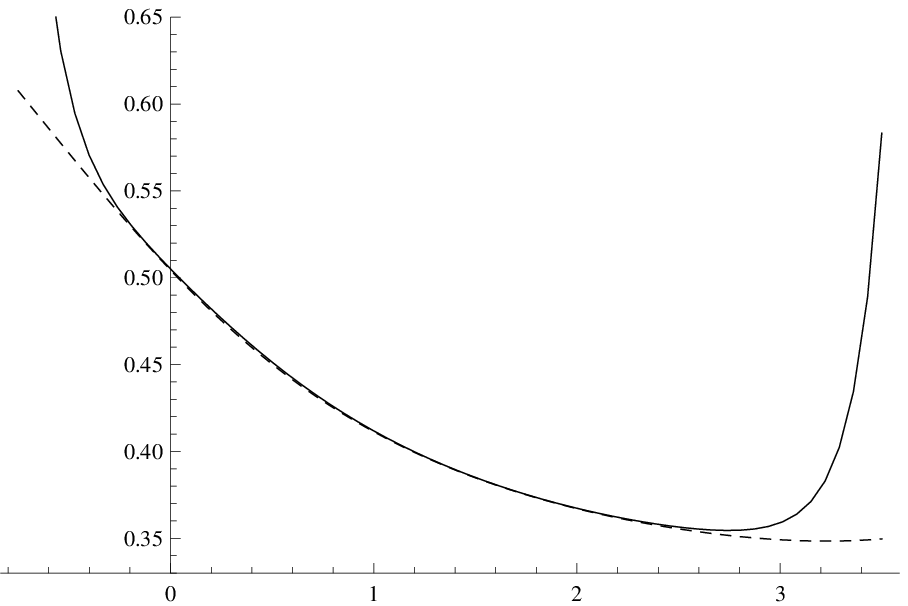} \\ \hline 
$n=4$ & $n=5$ \\
\includegraphics[width=.465\textwidth,height=.24\textheight]{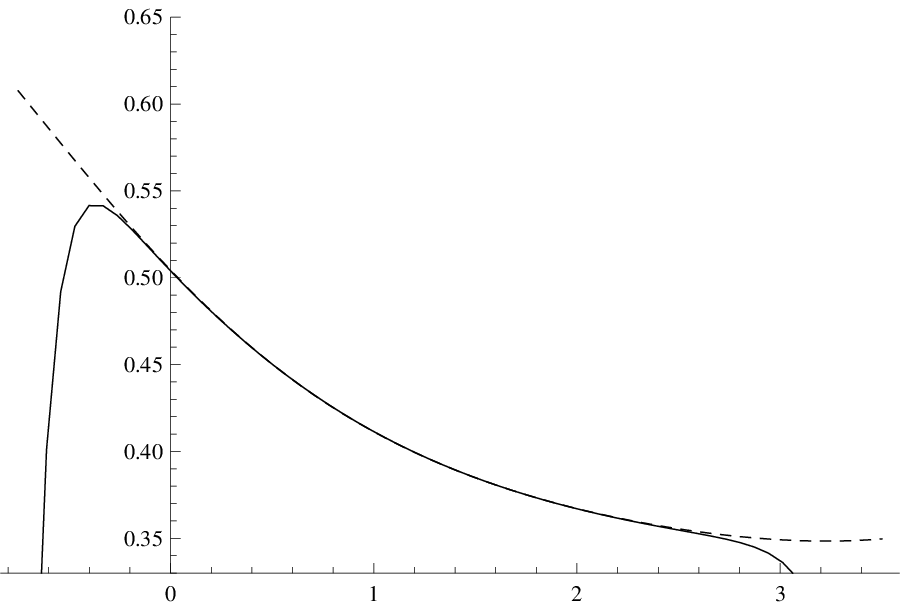} &
\includegraphics[width=.465\textwidth,height=.24\textheight]{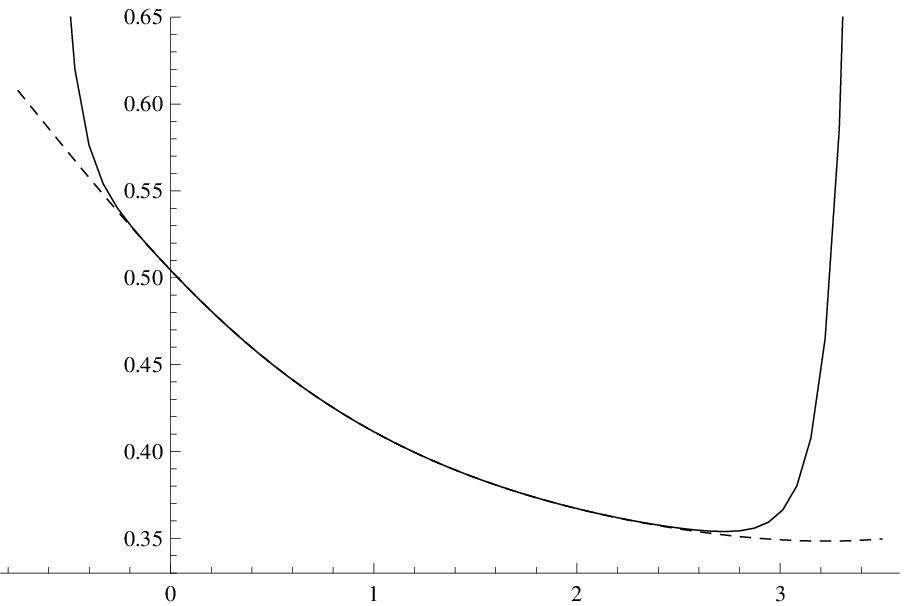} \\ \hline
\end{tabular}
\begin{align}
\nu_0(dz)
        &=              0 , &
\nu_1(dz)
        &=              \frac{1}{\sqrt{ 2 \pi s_1^2}} \exp \( \frac{-(z-m_1)^2}{2 s_i^2} \) dz .
\end{align}
\caption{We plot $\sig^{(n)}$, the order $\Oc(\eps^n)$ approximation of implied volatility (solid black), and $\sig^{\eps}$, the exact implied volatility (dashed black) as a function of $\text{LMMR}$.  The following parameters are used in these plots: $a_0=0.30$, $a_1=0.00$, $c_0=0.00$, $c_1=0.00$, $s_1=0.2$, $m_1=-0.40$, $\eps=4$, $\beta=-1.25$, $t=0.125$, $y=0.10$.}
\label{fig:impvol2}
\end{figure}


\clearpage
\begin{figure}
\centering
\begin{tabular}{ | c | c | c | }
\hline
87 DTM & 115 DTM & 142 DTM\\
\includegraphics[width=.305\textwidth,height=.18\textheight]{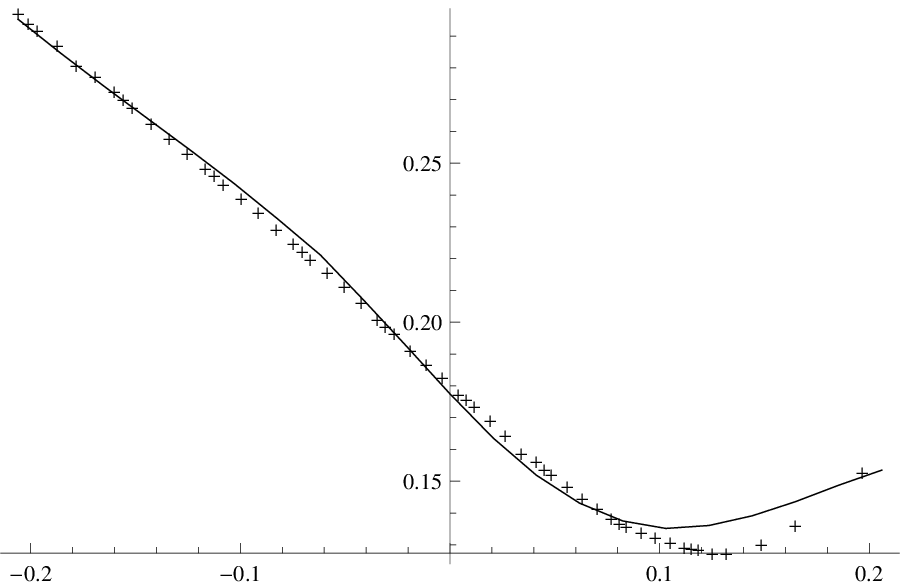} &
\includegraphics[width=.305\textwidth,height=.18\textheight]{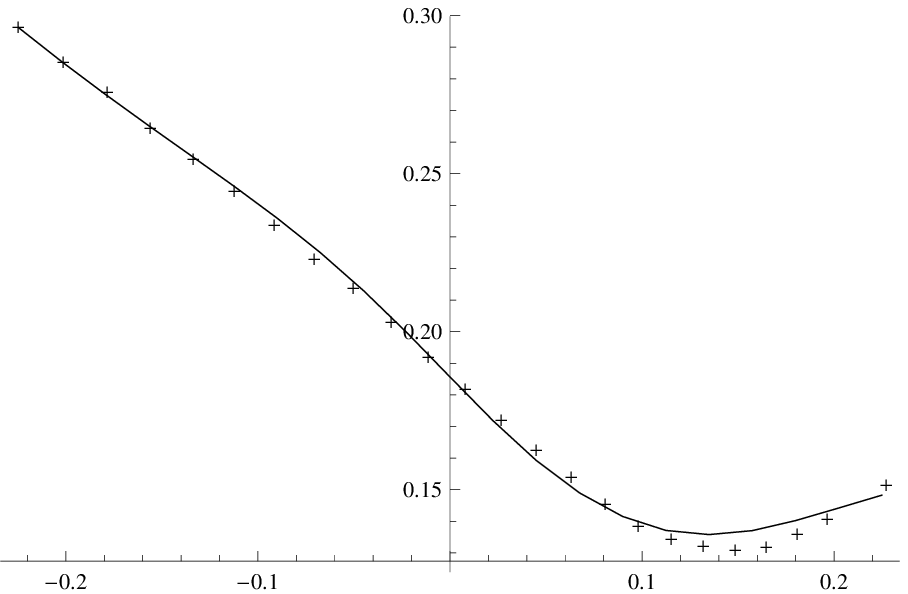} &
\includegraphics[width=.305\textwidth,height=.18\textheight]{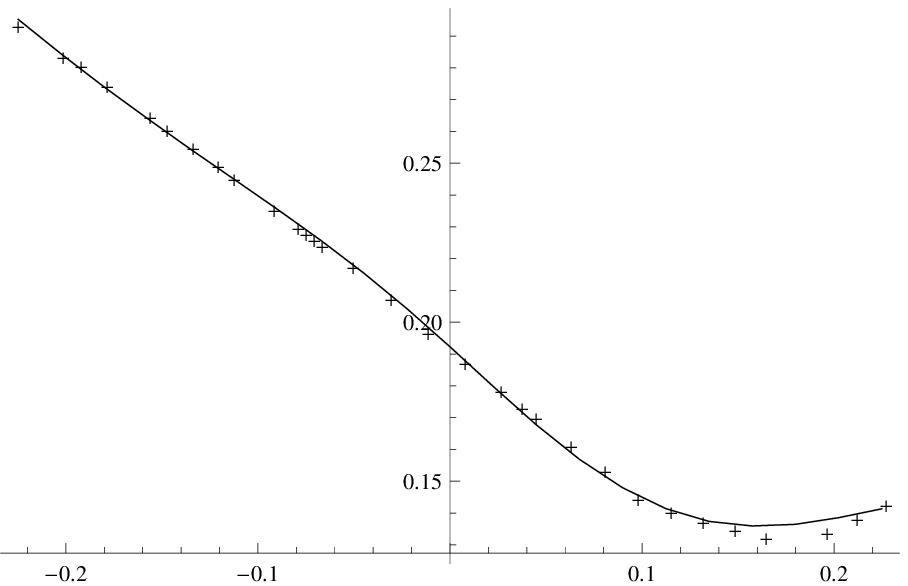} \\
\hline
\end{tabular}
\caption{Using the CEV-like model from Section \ref{sec:example}, we perform a calibration to S{\&}P500 options from January 24, 2012.  
The horizontal axis is in units of $\log$-moneyness: $\text{LM}:= k - y$ and
the vertical axis in units of implied volatility.  
The fit is a least-squares algorithm to implied volatilities across the three maturities.  
We emphasize that \emph{we do not fit maturity-by-maturity}.  
The L\'evy measures $\nu_0$ and $\nu_1$ are Gaussian with common mean $m$ and variance $s$ 
but different intensities $\Gamma_0$ and $\Gamma_1$.  
Thus, we allow the jump intensity, but not the jump distribution, to change as a function of $y$.  
The parameters resulting from the calibration are as follows: $a_0=0.059$, $c_0=0.009$, 
$\Gamma_0=1.105$, $a_1=0.057$, $c_1=0.010$, $\Gamma_1=1.095$, $m=-0.076$, $s=0.078$, $\beta=0.410 $, $\eps=1.00$.  
Without loss of generality, we assume $y=0$, which simply results in a rescaling of parameters.}
\label{fig:SP500}
\end{figure}

\begin{table}
\centering
\begin{tabular}{c|c|cccccccccc}
\hline
$N$ & $T_N/T_0$ & {} & {} & {} & {} & $\text{IV}^{(N)}$ & {} & {} & {} & {} & {} \\ 
\hline
0 & 1.00 & 0.2420 & 0.2162 & 0.1933 & 0.1719 & 0.1486 & 0.1222 & 0.1014 & 0.0929 & 0.0963 & 0.1046 \\ 
1 & 1.04 & 0.2929 & 0.2683 & 0.2476 & 0.2306 & 0.2166 & 0.2006 & 0.1676 & 0.1318 & 0.1183 & 0.1211 \\ 
2 & 1.49 & 0.2960 & 0.2709 & 0.2479 & 0.2265 & 0.2049 & 0.1841 & 0.1743 & 0.1558 & 0.1341 & 0.1307 \\ 
3 & 2.22 & 0.2951 & 0.2698 & 0.2475 & 0.2276 & 0.2088 & 0.1887 & 0.1634 & 0.1547 & 0.1429 & 0.1354 \\ 
4 & 3.26 & 0.2953 & 0.2701 & 0.2475 & 0.2272 & 0.2077 & 0.1877 & 0.1694 & 0.1483 & 0.1437 & 0.1379 \\ 
5 & 4.48 & 0.2952 & 0.2700 & 0.2475 & 0.2273 & 0.2079 & 0.1879 & 0.1674 & 0.1518 & 0.1404 & 0.1390 \\ 
6 & 6.16 & 0.2952 & 0.2700 & 0.2475 & 0.2273 & 0.2080 & 0.1878 & 0.1675 & 0.1519 & 0.1403 & 0.1391 \\ 
\hline
{} & LM &  -0.225 & -0.180 & -0.135 & -0.090 & -0.045 & 0.000 & 0.045 & 0.090 & 0.135 & 0.180  \\ 
\hline
\end{tabular}
\caption{Using the parameters obtained in the calibration to S{\&}P500 options (see Figure \ref{fig:SP500}) we compute approximate prices $u^{(N)}$ using equation (\ref{eq:price.CEV}).  We then compute implied volatilities ($\text{IV}^{(N)}$) by inverting the Black-Scholes pricing formula numerically.  We denote by $T_N$ the computational time required to compute implied volatilities for a series of strikes (listed above in unites of $\log$-moneyness: $\text{LM}:=k-y$) with a time to maturity of 142 days.   
Note that $T_0$ corresponds to the time it takes to compute IV's for an exponential L\'evy model.}
\label{tab:IV}
\end{table}


\clearpage
\begin{figure}
\centering
\includegraphics[width=.95\textwidth,height=.6\textheight]{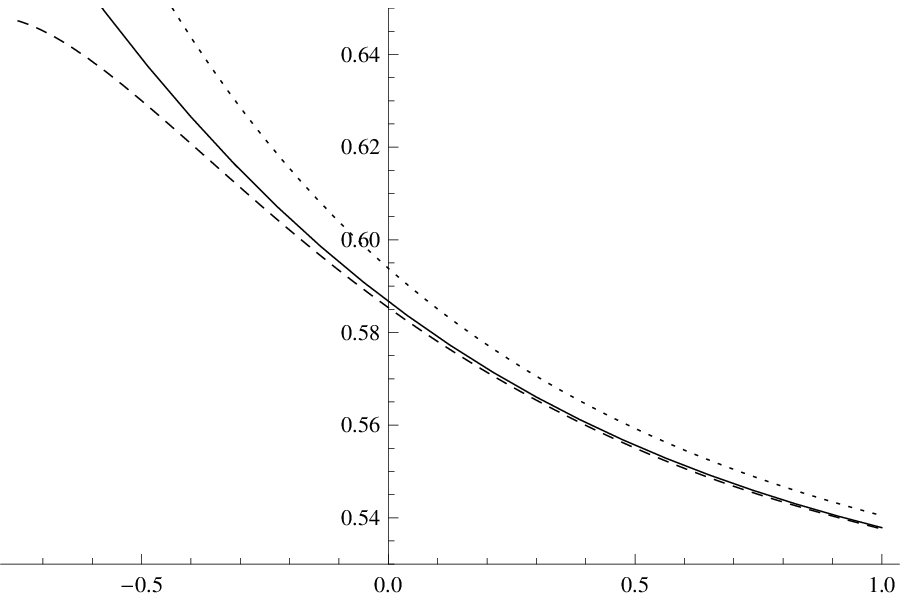} 
\caption{
We consider here the model proposed in Section~\ref{sec:impvol.2}.
We plot the exact implied volatility $\sig^\eps$ (solid), as well as the approximations $\sig^{(2,M)}$ (dashed) and $\sig^{(1,M)}$ (dotted) as a function of $\log$-moneyness: $\text{LM}:= k - y$.  
In the above plot we use the following parameters: 
$t=0.5$, $y=0.0$, $\beta=-2.0$, $\eps=1.0$, $a_0=0.5$ and $a_1=0.3$.  
Observe that $\sig^{(2,M)}$ closely approximates $\sig^\eps$ for all $\text{LM}>-0.5$.}
\label{fig:IV}
\end{figure}

\end{document}